\documentclass[11pt, a4paper, onecolumn]{quantumarticle}

\pdfoutput=1

\usepackage[T1]{fontenc} 
\usepackage[english]{babel} 
\usepackage[utf8]{inputenc}
\usepackage{csquotes}
\usepackage[most]{tcolorbox}
\usepackage[table,dvipsnames]{xcolor}
\usepackage[psdextra]{hyperref}
\usepackage{bookmark}
\usepackage{amsmath,amssymb,amsfonts,mathtools,amsthm,comment}
\usepackage{graphicx,tikz}
\usepackage{bbm}
\usepackage{physics}
\usepackage{tabularray}
\usepackage[capitalise]{cleveref}
\usepackage[ruled,boxed,linesnumbered]{algorithm2e}
\usepackage[margin=1in]{geometry}
\usepackage{tikz}
\usepackage{tikz}
\usetikzlibrary{shapes.geometric}
\usepackage{amsmath}
\usepackage{subcaption}

\usepackage{stmaryrd}

\usepackage[
style=numeric,
sorting=none,
maxbibnames=10
]{biblatex}
\AtEveryBibitem{
    \clearfield{urlyear}
    \clearfield{urlmonth}
    \clearfield{month}
}

\usepackage{enumitem}
\usepackage[normalem]{ulem}
\usepackage{multirow}
\usepackage{parskip}

\usepackage{xr}
\newtcolorbox[auto counter]{iobox}[2][]{
  enhanced, breakable,
  colframe=blue!50,
  colback=white,
  coltitle=black,
  colbacktitle=white,        
  fonttitle=\bfseries,
  title={Problem~\thetcbcounter: #2},   
  label={#1},
  titlerule=0pt,
  boxrule=0.8pt, arc=3pt,
  boxsep=0pt,
  toptitle=2pt, bottomtitle=2pt,
  top=0pt, bottom=6pt, left=6pt, right=6pt,
  subtitle style={
    boxrule=0pt,
    colback=blue!12,
    coltext=black,
    fontupper=\bfseries,     
    boxsep=0pt,
    top=2pt, bottom=2pt, left=6pt, right=6pt,
    before skip=0pt,         
    after skip=4pt,          
  },
}
\usepackage{authblk}
\newtheorem{theorem}{Theorem}
\newtheorem*{theorem*}{Theorem}
\newtheorem{corollary}[theorem]{Corollary}
\newtheorem{lemma}[theorem]{Lemma}

\newtheorem{definition}[theorem]{Definition}
\theoremstyle{definition}

\newcommand{\fro}{\mathrm{F}}

\newcommand{\y}{\mathrm{y}}

\newcommand{\myleft}{\mathopen{}\mathclose\bgroup\left}
\newcommand{\myright}{\aftergroup\egroup\right}

\DeclarePairedDelimiterX{\iiiNorm}[1]{\lvert}{\rvert}{%
  \delimsize\lvert\delimsize\lvert#1\delimsize\rvert\delimsize\rvert%
}

\DeclarePairedDelimiterXPP\snorm[1]{}\lVert\rVert{_\infty}{\ifblank{#1}{\,\cdot\,}{#1}}   

\DeclarePairedDelimiterXPP\twonorm[1]{}\lVert\rVert{_2}{\ifblank{#1}{\,\cdot\,}{#1}}   

\DeclarePairedDelimiterXPP\trnorm[1]{}\lVert\rVert{_1}{\ifblank{#1}{\,\cdot\,}{#1}}   

\DeclarePairedDelimiterXPP\fnorm[1]{}\lVert\rVert{_{\fro}}{\ifblank{#1}{\,\cdot\,}{#1}}   

\DeclarePairedDelimiterXPP\dnorm[1]{}\lVert\rVert{_\diamond}{\ifblank{#1}{\,\cdot\,}{#1}}   

\DeclarePairedDelimiterXPP\cbnorm[1]{}\lVert\rVert{_\mathrm{cb}}{\ifblank{#1}{\,\cdot\,}{#1}}   
\DeclarePairedDelimiterXPP\onenorm[1]{}\lVert\rVert{_{1\rightarrow 1}}{\ifblank{#1}{\,\cdot\,}{#1}}   
\DeclarePairedDelimiterXPP\ddnorm[1]{}\lVert\rVert{_{\diamond\rightarrow \diamond}}{\ifblank{#1}{\,\cdot\,}{#1}}   
\DeclarePairedDelimiterXPP\ssnorm[1]{}\lVert\rVert{_{\infty\rightarrow\infty}}{\ifblank{#1}{\,\cdot\,}{#1}}   

\DeclarePairedDelimiterX\Set[1]\{\}{%
  
  #1
}

\DeclarePairedDelimiterX\innerp[2]{\langle}{\rangle}{%
  \ifblank{#1}{\,\cdot\,}{#1} , \ifblank{#2}{\,\cdot\,}{#2}%
}

\DeclarePairedDelimiterX\sandwich[3]{\langle}{\rangle}%
  {#1\,\delimsize\vert\kern0.15ex\mathopen{}#2\kern0.15ex\delimsize\vert\kern0.15ex\mathopen{}#3}

\DeclarePairedDelimiterX\obraket[2]{(}{)}%
  {#1\kern0.15ex\delimsize\vert\kern0.15ex\mathopen{}#2}

\DeclarePairedDelimiterX\oketbra[2]{\vert}{\vert}%
  {#1\kern0.15ex\delimsize)\delimsize(\kern0.15ex\mathopen{}#2}

\DeclarePairedDelimiterX\osandwich[3]{(}{)}%
  {#1\,\delimsize\vert\kern0.15ex\mathopen{}#2\kern0.15ex\delimsize\vert\kern0.15ex\mathopen{}#3}

\newcommand{\Hil}{\mathcal{H}}
\newcommand{\Scurvy}{\mathcal{S}}
\newcommand{\Ncurvy}{\mathcal{N}}
\newcommand{\symplspace}{\mathbb{F}_2^{2n}}
\newcommand{\Pn}{\mathcal{P}_n}

\newcommand{\supp}{\mathrm{support}}

\newcommand{\Ftwotwon}{\mathbb{F}_2^{2n}}

\newcommand{\blfootnote}[1]{%
  \begingroup
  \renewcommand\thefootnote{}\footnote{#1}%
  \addtocounter{footnote}{-1}%
  \endgroup
}
\makeatletter
\newcommand\rst@title{}
\newcommand\rst@key{}
\newcommand\rst@env{}
\newcommand\rst@ctr{}

\newcommand\newrestated[3]{%
  \newenvironment{#1}[1][]{%
    \renewcommand\rst@env{#2}%
    \renewcommand\rst@ctr{#3}%
    \renewcommand\rst@title{##1}%
    \@ifnextchar\label\rst@grab{\rst@start{}}%
  }{%
    \end{#2}\addtocounter{#3}{-1}%
  }%
}
\def\rst@grab\label#1{\rst@start{#1}}
\newcommand\rst@start[1]{%
  \renewcommand\rst@key{#1}%
  \ifx\rst@key\@empty\else
    \expandafter\renewcommand\csname the\rst@ctr\endcsname{\ref*{#1}}%
  \fi
  \ifx\rst@title\@empty
    \expandafter\begin\expandafter{\rst@env}%
  \else
    \expandafter\begin\expandafter{\rst@env}[\rst@title]%
  \fi
}
\makeatother

\newrestated{restatedtheorem}{theorem}{theorem}
\newrestated{restatedlemma}{lemma}{theorem}

\makeatletter
\def\thm@space@setup{%
  \thm@preskip=\medskipamount
  \thm@postskip=\medskipamount
}
\makeatother

\title{Logical information localisation in stabiliser codes via single-qubit measurements}
\date{March 2026}

\begin{document}
\author[1]{Jelena Mackeprang$^*\,$}
\author[1,2]{Hemant Sharma$^*\,$}

\author[1]{Jonas Helsen}
\affil[1]{QuSoft and CWI, Science Park 123, 1098XG Amsterdam, The Netherlands}
\affil[2]{QuTech, TU Delft, Delft, The Netherlands}
\blfootnote{{*}These authors contributed equally to this work.}

\maketitle
\begin{abstract}
 
Stabiliser path finding (SPF) has previously been introduced as a method to localise logical information in a stabiliser code undergoing loss onto a single pre-specified target qubit, using only one round of single-qubit measurements. When working with limited resources and flying qubits, this fast read-out of logical information is a helpful tool for fault-tolerant communication. In this work, we provide a broad analytical and computational study of localisation via SPF. We introduce $g$-SPF, where the task is to localise the logical information onto a set of at most $g$ unspecified target qubits. Through analytical arguments based on percolation theory and the disjointness of stabiliser codes, we prove that for i.i.d. qubit loss with probability $p<1/2$ and sufficiently large planar surface codes, localisation via $g$-SPF for constant $g$ succeeds with a probability converging to one, which establishes a localisation threshold. Furthermore, we propose and implement two algorithms to solve SPF. The first is exact and formulates SPF as an integer linear program, whereas the second is heuristic and formulates SPF as a decoding problem. We validate both algorithms by numerically reproducing the localisation threshold for the surface code and demonstrate that the heuristic algorithm is considerably faster. Our work significantly reduces the time required to solve SPF compared to current state-of-the-art algorithms, allowing us to study localisation in substantially larger stabiliser codes than previously considered in the literature. Together, these theoretical and computational results open the door to various applications, such as fault-tolerant teleportation and efficient logical fusion.

\end{abstract}
\section{Introduction}\label{sec:intro}

Measurement-based quantum computation (MBQC) and communication promise various different future applications, such as secure communication, distributed computing and cryptographic protocols~\cite{quantum_internet,zhang2025towards,escolafarras2025quantum,BARRAL2025100747,Krenn2016}. Photonic hardware is a promising route towards scalable, fast and energy-efficient quantum computation~\cite{zhu2026recent,slussarenko2019photonic}. Photons are naturally suited for MBQC where any computation is performed using only (adaptive) single-qubit measurements and local Clifford operations on a large entangled resource state~\cite{wei2019measurement,PhysRevA.76.052315,PhysRevLett.86.5188,maring2024versatile}. They are also the most natural resource used in fusion based quantum computation~\cite{bartolucci2023fusion} and of course in quantum communication~\cite{quantumnetworks,krenn2016quantum}. However, quantum computation or communication with photons faces the problem of photon loss. Fault-tolerant schemes and suitable quantum error correcting codes for MBQC  , quantum communication or fusion based quantum computation are therefore highly sought after~\cite{abughanem2026toward,PRXQuantum.6.020304,bartolucci2025comparisonschemeshighlyloss,encodedfusion,qj8s-j274,thomas2024fusion}.

Bell et al.~\cite{bellOptimizingGraphCodes2023} searched for graph codes - stabiliser codes defined with respect to graphs~\cite{schlingemannStabilizerCodesCan2001, Hwang_2016}- suitable for their proposed schemes for photonic quantum repeaters and fusion-based quantum computation. Their search relies on the stabiliser path finding (SPF) method of Morley-Short et al.~\cite{morley-shortLosstolerantTeleportationLarge2019,ltdecode}, where the goal is to find a pair of representatives for logical $X$ and $Z$ that only anti-commute qubit-wise on one designated target qubit and have no support on any lost qubits. Measuring these representatives qubit-wise \emph{localises} the logical information onto the target qubit. This procedure constitutes a local operations and classical communication (LOCC) protocol~\cite{chitambarEverythingYouAlways2014} that localises information onto one qubit in a stabiliser code. Bell et al. propose to use SPF for \emph{adaptive fusions} (see Ref.~\cite{photonic_fusion}) tolerant to fusion failures and qubit losses. 
The SPF method from Ref.~\cite{morley-shortLosstolerantTeleportationLarge2019} works by exhaustively listing logical representatives of the code suitable for SPF, which becomes infeasible for larger codes due to the exponentially increasing number of logical representatives. In practice, this limited the search in Ref.~\cite{bellOptimizingGraphCodes2023} to twelve-qubit graphs on a computing cluster. The procedure in Ref.~\cite{bellOptimizingGraphCodes2023} is also purely numerical and provides no analytical guarantee that loss tolerance persists at larger sizes. Furthermore, in Ref.~\cite{morley-shortLosstolerantTeleportationLarge2019} it is assumed that the target qubit is never lost. Arbitrarily choosing a particular target qubit may lead to sub-optimal success rates of the fault-tolerant scheme under consideration. Due to the broad applications of SPF suggested in Ref.~\cite{bellOptimizingGraphCodes2023} it is desirable to find scalable algorithms to solve it. Moreover, localisation via SPF is an interesting theoretical concept on its own that has - to our knowledge - not been studied in this form in previous literature. 
  
In this work, we thus provide a first theoretical and numerical analysis of localisation via SPF. 
We introduce another version of SPF, dubbed `$g$-SPF', by relaxing the condition of qubit-wise commutation of $X$ and $Z$ on all but a designated target qubit to qubit-wise commutation on all but an arbitrary subset of at most $g$ qubits. This is a more natural definition of the problem statement.  
We analyse $g$-SPF theoretically and show that the planar surface code on a square lattice exhibits a threshold for localisation via $g$-SPF for constant $g$ at an i.i.d. loss probability $p=1/2$, which, interestingly, coincides with its standard erasure threshold - despite localisation and erasure decoding not necessarily being equivalent.  
Moreover, we propose and implement a fast deterministic algorithm for SPF both for its original formulation with a designated target qubit and the more relaxed version $g$-SPF. Additionally, we propose and implement a fast heuristic algorithm for $g$-SPF. We use the linear structure of quantum error correcting codes to encode the loss of qubits and treat SPF as an optimisation problem turned either into a deterministic integer linear program or a heuristic decoding problem. Our algorithms are orders of magnitude faster than the one suggested in Ref.~\cite{morley-shortLosstolerantTeleportationLarge2019}, which enables the systematic search for large codes suitable for fault-tolerant schemes. 

This paper is structured as follows. In Section~\ref{sec:results} we give a brief informal overview of our analytical as well as computational results. Only basic knowledge of stabiliser codes is required to understand this section. Section~\ref{sec:background} reviews the preliminaries required to understand this work and introduces our notation. In Section~\ref{sec:formalisation}, we discuss stabiliser path finding and formalise it as an optimisation problem. We then move on to our own contributions. In Section~\ref{sec:threshold_surfacecode} we prove that the planar surface code exhibits a localisation threshold of $p=1/2$, where $p$ denotes the qubit loss probability. Our computational contributions start with Section~\ref{sec:opt_algos}, where we introduce our deterministic and heuristic algorithms D-LoFi and H-LoFi that solve SPF. We present our numerical results in Section~\ref{sec:numerics}. In the end, we give a conclusion and outlook in Section~\ref{sec:concl}.

\section{Results}\label{sec:results}
Here, we informally sketch our main analytical and numerical results. We consider both the original formulation of SPF, where the goal is to find an $X$ and a $Z$ with no support on the lost qubits that anti-commute qubit-wise on only a designated target qubit and commute qubit-wise elsewhere and our newly introduced version $g$-SPF, where $X$ and $Z$ must anti-commute on at most $g$ arbitrary target qubits. In the following, we will call any pair $(X,Z)$ with no support on the lost qubits that fulfil the specific anti-commutation conditions a `Target-SPF-' or `$g$-SPF pair'.

\subsection{Analytics}

  We analyse $g$-SPF analytically using the planar surface code as an example. Our main analytical result is Theorem~\ref{thm:threshold_surface_informal}, which we will state below. For a given i.i.d. qubit loss probability $p$, we define the success rate of localisation via $g$-SPF as the rate at which a $g$-SPF pair exists. A \emph{family} of stabiliser codes is a sequence $\{\mathcal{S}_n\}_{n \in \mathbb{N}}$ of codes defined by a common construction and indexed by a size parameter $n$, which here denotes the number of qubits. Theorem~\ref{thm:threshold_surface_informal} states that for $p<1/2$ the localisation success rate converges to one for $n \to \infty$ for the family of planar surface code.  
\begin{theorem}[$g$-SPF-threshold for the planar surface code - informal]\label{thm:threshold_surface_informal}
    Let $\{\Scurvy_n\}$ be a family of planar surface codes and let each qubit in the code be lost with a probability $p$. For any $p<1/2$, there exists a constant $g$ such that for large enough $n$ one can - with probability converging to one - find a $g$-SPF pair $(X,Z)$. That is, for $p<1/2$ the $g$-SPF success rate converges to one. Moreover, for any $p>1/2$ the $g$-SPF success rate converges to zero, meaning $p=1/2$ is the exact threshold.
\end{theorem}
Theorem~\ref{thm:threshold_surface_informal} states that localisation via LOCC onto a constant number of qubits on the surface code has an actual threshold. Moreover, this localisation threshold coincides with the standard erasure threshold from Ref.~\cite{staceThresholdsTopologicalCodes2009}, which is not trivially obvious. Practically speaking, due to its high threshold, the planar surface code turns out to be very suitable for localisation via $g$-SPF. Provided that the hardware allows for the creation of a planar surface code, one could thus use it to transmit quantum information over highly lossy channels. As any stabiliser code is local-Clifford equivalent to a graph code (see Ref.~\cite{graph_code_equiv_stabiliser}), one could consider the graph code version of the surface code for the tasks listed in Ref.~\cite{bellOptimizingGraphCodes2023}. 
To prove Theorem~\ref{thm:threshold_surface_informal}, we derive a similar lemma to the scrubbing lemma from Ref.~\cite{jochym-oconnorDisjointnessStabilizerCodes2018}, which we will state below. The support of an operator $P$ is the subset of qubits onto which $P$ acts non-trivially. We write $\supp(P)$ for the support of $P$. 
\begin{lemma}[Upper bound on minimum intersection size - informal]\label{lem:intersection_disjointness}
Let $\{\Scurvy_n\}$ be a family of codes encoding one logical qubit. Let $X$ and $Z$ be the logical operators of $\Scurvy_n$. Suppose that for each $\Scurvy_n$ one can find subsets $X_R$ and $Z_R$ of the set of all representatives for $X$ and $Z$ respectively consisting only of of pairwise disjoint operators, where $|X_R|=\Delta_X$ and $|Z_R|=\Delta_Z$. Then
\begin{equation}\label{eq:disjointness_bound_intersection}
    \min_{X \in X_R, Z\in Z_R} \big|\mathrm{supp}(Z) \cap \mathrm{supp}(X)\big| \leq n/(\Delta_X\Delta_Z).
\end{equation}
\end{lemma}
Lemma~\ref{lem:intersection_disjointness} upper bounds the minimum intersection size of $X$ and $Z$ representatives. The more pairwise disjoint representatives $X$ and $Z$ one can find, the smaller the minimum intersection becomes. Using known statements from percolation theory, we will show in Section~\ref{sec:threshold_surfacecode} that even after loss the number of \emph{surviving} pairwise disjoint representatives for $X$ and $Z$ scales with $\sqrt{n}$, provided that the single-qubit loss probability $p$ is smaller than 1/2. This, together with Lemma~\ref{lem:intersection_disjointness}, then implies a constant minimum intersection between \emph{surviving} $X$ and $Z$ representatives.
\subsection{Algorithms}
 
The main computational contribution of our work is the proposal of a deterministic algorithm for Target-SPF and $g$-SPF and a heuristic one for $g$-SPF. In the first algorithm, Deterministic Localisation Finder (D-LoFi), we encode Target- and $g$-SPF as constrained quadratic optimisation problems, which we formulate as integer linear programs (ILPs) and solve using an ILP solver. 
Moreover, we minimise the size of the total support of the SPF pair for a given set of lost qubits (referred to as a loss configuration), so that the number of required single-qubit measurements for localisation is minimal.
Our second algorithm, Heuristic Localisation Finder (H-LoFi), encodes $g$-SPF as a type of most-likely-error decoding problem that we solve using a decoder. 
It first formulates the search for a short $Z$ representative as a decoding problem. Then, it formulates the search for an $X$ representative that anti-commutes qubit-wise with $Z$ on a small set of qubits as another decoding problem. We use the stabiliser tableau representation of stabiliser codes~\cite{gottesmanStabilizerCodesQuantum1997} and utilise linear algebra on symplectic spaces to encode the constraint that the SPF pair $(X,Z)$ cannot have support on the lost qubits. That is, our pre-processing maps the input stabiliser tableau to one describing the updated stabiliser code after loss. Note that the tableau representation is the natural representation of stabiliser codes and does not require any particular structure of the code, allowing us to analyse stabiliser codes beyond graph codes. 

We numerically reproduce threshold behaviour similar to what was predicted by Theorem~\ref{thm:threshold_surface_informal} using our algorithms D-LoFi and H-LoFi in Fig.~\ref{fig: sc_thres_informal}, where we plot the success rate of $g$-SPF on the rotated surface code for increasing loss probability $p$, where $g=1$. Both our algorithms indicate that a stronger version of Theorem~\ref{thm:threshold_surface_informal} holds, namely that $g$ can be set to one, though an analytical proof is needed to make a definite statement. The deterministic algorithm D-LoFi is exact, meaning it will always find a $g$-SPF pair if it exists. Fig.~\ref{fig: sc_thres_informal}(b) demonstrates that our H-LoFi algorithm closely approximates the exact success rates obtained by D-LoFi. We conclude that both algorithms serve as a reliable tool to faithfully compute thresholds for $g$-SPF on general stabiliser codes.

\begin{figure*}[t]
    \begin{minipage}{0.49\textwidth}
        \includegraphics[width=0.95\linewidth]{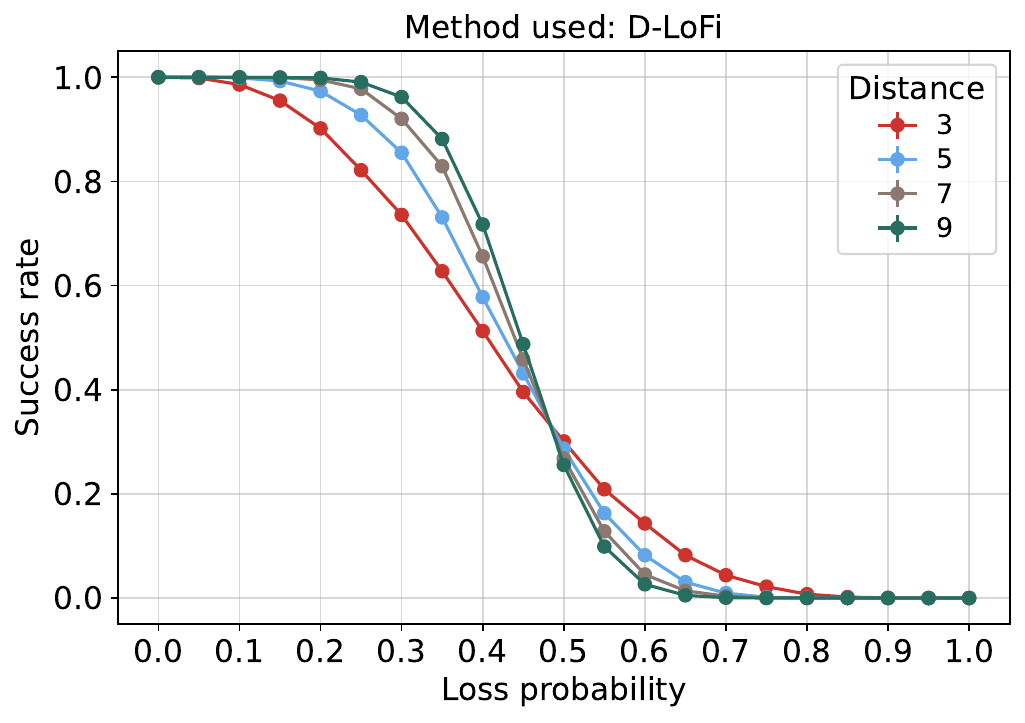}
        \subcaption{Method: Deterministic Localisation Finder (D-LoFi)}
    \end{minipage}
    \begin{minipage}{0.49\textwidth}
        \includegraphics[width=0.95\linewidth]{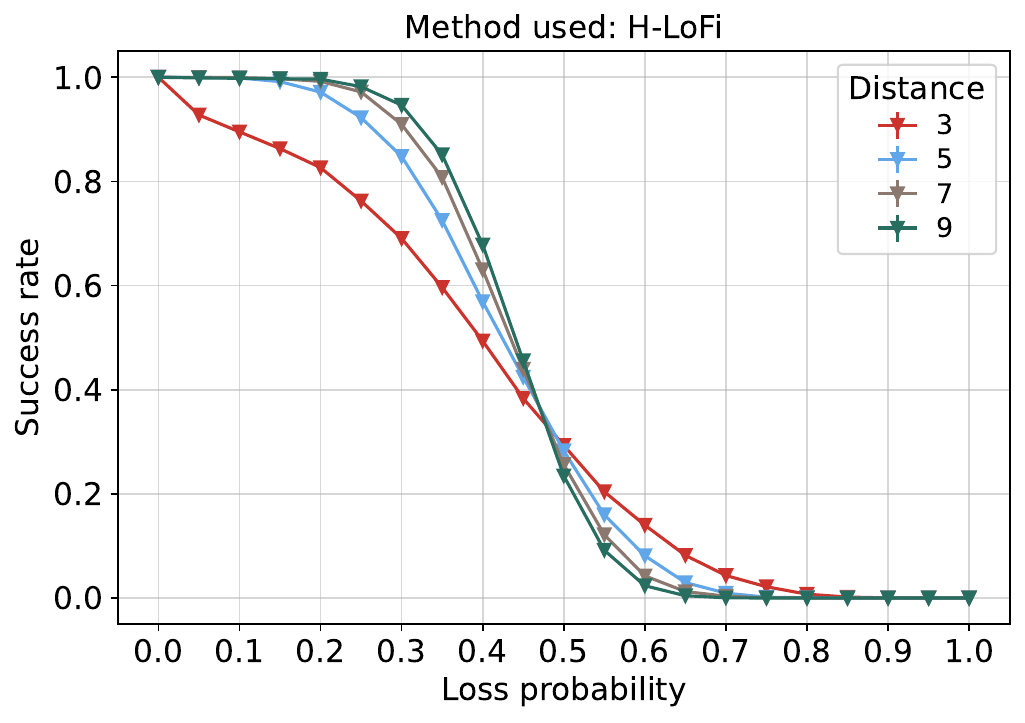}
        \subcaption{Method: Heuristic Localisation Finder (H-LoFi)}
    \end{minipage}
    \caption{Success rate at which a $g$-SPF pair is found for the rotated surface code (see for example Ref.~\cite{rotated_surface_code}) and $g=1$. For every loss probability $p$ we Monte Carlo sample at least 5000 loss configurations and estimate the success rate of $g$-SPF.}
    \label{fig: sc_thres_informal}
\end{figure*}

Fig.~\ref{fig: sc_thres_speed_informal}(a) compares the average runtime of H-LoFi with the average runtime of D-LoFi for a fixed loss probability $p=0.1$. We observe that H-LoFi is orders of magnitude faster than D-LoFi. Moreover, Fig.~\ref{fig: sc_thres_speed_informal}(b) compares the relative increase of the total support size $|\supp(X) \cup \supp(Z)|$ of the $g$-SPF pairs $(X,Z)$ found by the two algorithms. It shows that while H-LoFi is considerably faster than D-LoFi, we do not sacrifice much in terms of its performance. D-LoFi returns a $g$-SPF pair $(X,Z)$ that requires a minimal number of single-qubit measurements for subsequent localisation and the solution found by H-LoFi does not significantly increase this required number of measurements. 

\begin{figure*}[t]
 \begin{minipage}{0.49\textwidth}
        \includegraphics[width=0.95\linewidth]{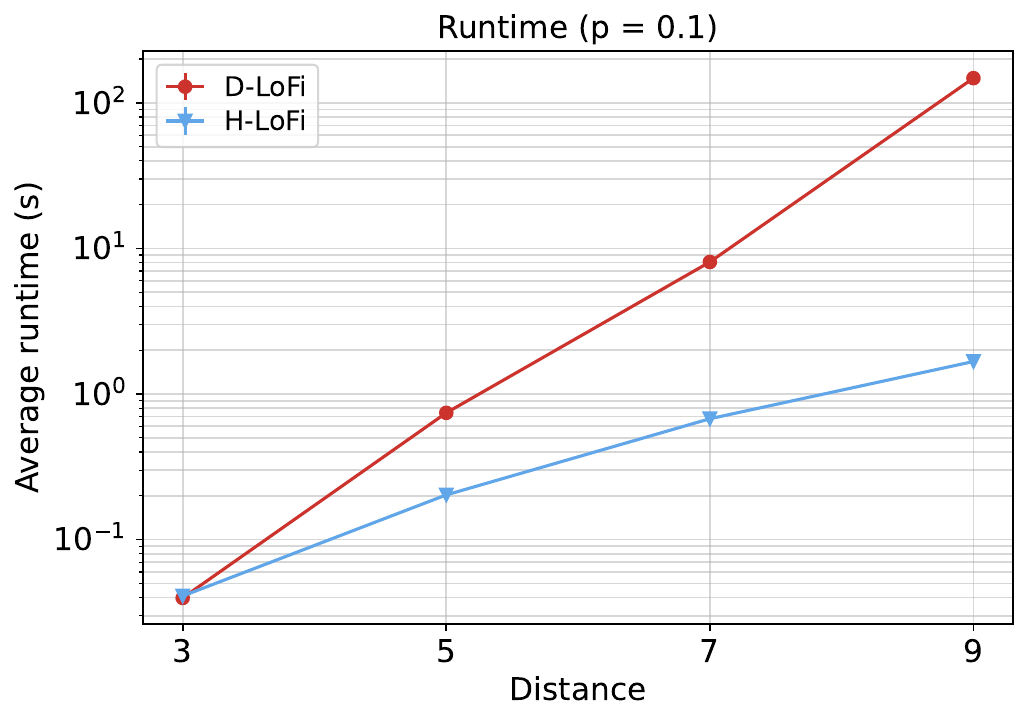}
        \subcaption{}
    \end{minipage}
    \begin{minipage}{0.49\textwidth}
        \includegraphics[width=0.95\linewidth]{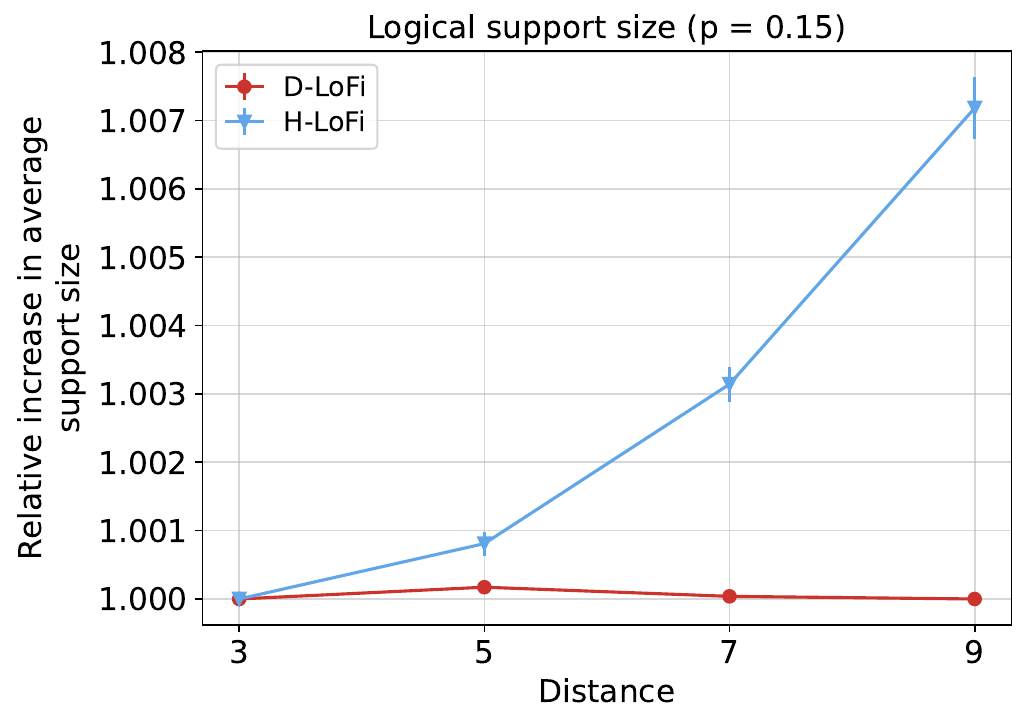}
        \subcaption{}
        \end{minipage}
    \caption{In (a) we plot the average runtime of D-LoFi and H-LoFi with increasing code distance for a fixed single-qubit loss probability $p=0.1$. We average the runtime over 5000 Monte Carlo samples taken at $p=0.1$ during the experiment to obtain the $g$-SPF thresholds in Fig.~\ref{fig: sc_thres_informal}. In (b) we plot the average relative increase of the support of the $g$-SPF pair $(X,Z)$ that D-LoFi and H-LoFi found, for the fixed loss probability $p=0.15$.} 
    \label{fig: sc_thres_speed_informal}
\end{figure*}
 
We compare the runtime of our deterministic algorithm D-LoFi to the runtime of the algorithm proposed in Ref.~\cite{morley-shortLosstolerantTeleportationLarge2019} in Fig.~\ref{fig: ms_rt_comp_informal}, using the `crazy graph' code as an example. (See Section~\ref{sec:graph_codes} for details.) The runtimes of the algorithm in Ref.~\cite{morley-shortLosstolerantTeleportationLarge2019} for more than 18 qubits exceeded our imposed wall time of about 16 hours. We compute the average runtime of D-LoFi by averaging over a large number of unique loss configurations. Fig.~\ref{fig: ms_rt_comp_informal} shows that our method is orders of magnitude faster than the one from Ref.~\cite{morley-shortLosstolerantTeleportationLarge2019}. Even though D-LoFi presumably scales exponentially due to the nature of the problem, the algorithm from Ref.~\cite{morley-shortLosstolerantTeleportationLarge2019} becomes infeasible much faster. One can therefore use our methods to compute thresholds for large codes or systematically search for codes suitable for the tasks discussed in Ref.~\cite{bellOptimizingGraphCodes2023}.

\begin{figure}[htbp]
\hspace{0.225\textwidth}
    \begin{minipage}{0.99\textwidth}
        \includegraphics[width=0.5\linewidth]{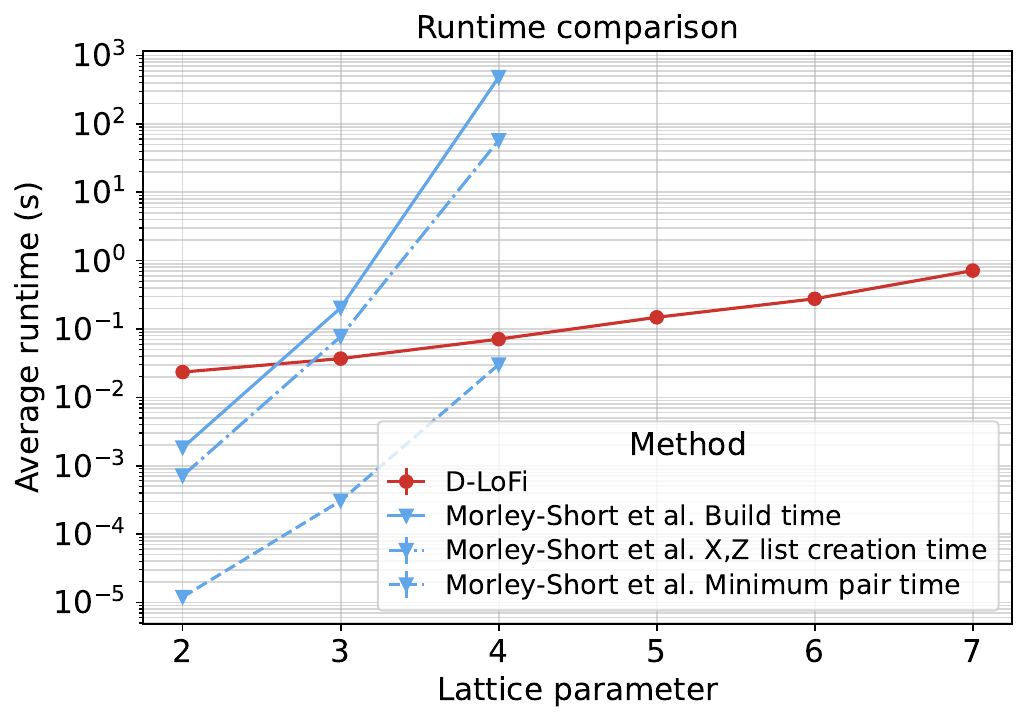}
    \end{minipage}
    \caption{We plot the runtime comparison of D-LoFi with the algorithm from Ref.~\cite{morley-shortLosstolerantTeleportationLarge2019} (labelled as \emph{Morley-Short et al.}). We sample either all possible loss configurations (for small codes) or 1000 unique ones for each lattice parameter and run both algorithms. Both algorithms find the minimal Target-SPF pair for one of the graph codes from Ref.~\cite{morley-shortLosstolerantTeleportationLarge2019}, parametrised by the lattice parameter $l$. The number of qubits $n$ is given by $n=l^2+2$. The algorithm from Ref.~\cite{morley-shortLosstolerantTeleportationLarge2019} consists of three steps, which we explain in more detail in Section~\ref{sec:numerics}. To find an SPF pair one needs to perform all steps; thus the total runtime is the sum of the three runtimes plotted here. We do not plot the runtimes of the algorithm from Ref.~\cite{morley-shortLosstolerantTeleportationLarge2019} for $l>4$ ($n>18$), as those exceed our imposed wall time limit of 60,000 seconds ($\sim$ 16 hours and 40 minutes). We compare with the average runtime of our algorithm D-LoFi, averaged over the same loss configurations.}
    \label{fig: ms_rt_comp_informal}
\end{figure}

As the runtime of the D-LoFi scales better than that of the algorithm from Ref.~\cite{morley-shortLosstolerantTeleportationLarge2019}, we can run D-LoFi for graph codes of significantly larger sizes than those considered in Ref.~\cite{morley-shortLosstolerantTeleportationLarge2019}, which we will discuss in more detail in Section~\ref{sec:numerics}.

\section{Preliminaries}\label{sec:background}

In this section, we will go over the basic preliminaries on stabiliser codes and their symplectic representation needed to understand this work. Then, we will define the stabiliser codes we consider here, which include graph codes and the planar surface code. We will also go over our notation.

 Here we do not go into any details on quantum error correction, but only give the required background on stabiliser codes. For an introduction to QEC, see for example Ref.~\cite{lidarQuantumErrorCorrection2013} or Ref.~\cite{gottesmanIntroductionQuantumError2009}. 
A stabiliser code is defined as the joint positive eigenspace of a stabiliser (group) $\Scurvy \subseteq \mathcal{P}_n$, which is an abelian subgroup of the $n$-qubit Pauli group ($\mathcal{P}_n$) such that $- \mathbf{I}_n \notin \Scurvy$, where $\mathbf{I}_n$ is the identity on $\Hil^{\otimes n}$. We can find a minimal generating set $\{g_i\}$ for $\Scurvy$. Throughout this paper we will write $r=|\{g_i\}|$ for the rank of $\Scurvy$. A stabiliser code defined on $n$ qubits with $r$ generators encodes $k:=n-r$ logical qubits. We will write $Q$ for the set of all qubits.
The logical space is determined by the normaliser of $\Scurvy$ in $\mathcal{P}_n$, i.e.
\begin{equation}
    \mathcal{N}(\Scurvy)=\{P \in \mathcal{P}_n| \, [P,s]=0 \,  \forall \,  s \in \Scurvy\}.
\end{equation}
In fact, the quotient group $\mathcal{N}(\Scurvy)/\Scurvy$ is isomorphic to $\mathcal{P}_k$. The elements of $\Ncurvy(\Scurvy)/\Scurvy$ are equivalence classes. When we want to refer to them as a set, i.e. when we refer to the \emph{set of all representatives} for a given $P \in \Ncurvy(\Scurvy)/\Scurvy$, we will write $[P]$. $P$ then is the symbol for the equivalence class. Note that the logical basis choice is arbitrary. The most important take-away from this section is that $[P]$ is a set of multiple equivalent representatives of the same logical operator. Therefore, logical information is encoded redundantly. For example, if one qubit is lost, one can often find some representative in $[P]$ that acts as the identity on this qubit and whose restriction onto all other preserved qubits is still a valid representative of $P$. This is the principle behind erasure decoding. 

We write $\Ncurvy(\Scurvy)/\Scurvy$ for the quotient set and $\Ncurvy(\Scurvy)-\Scurvy$ for the set difference.
$\mathcal{P}_n/\langle i \mathbf{I}_n \rangle$ is isomorphic to the symplectic space $\mathbb{F}_2^{2n}$, where the symplectic product is a bilinear form given by: %
\begin{equation} \label{eq:symplecticprod}
\langle \cdot , \cdot \rangle: \mathbb{F}_2^{2n} \times \mathbb{F}_2^{2n} \rightarrow \mathbb{F}_2:  \langle v , w \rangle= v \Omega w^T,  
\end{equation}
where $\Omega$ is a $2n \times 2n$ binary matrix:%
\begin{align}
    \Omega = \begin{bmatrix}
        0 & {I}_n\\
        {I}_n & 0
    \end{bmatrix}. \label{eq:Omega}
\end{align}
To identify each $v \in \Ftwotwon$ with a $P \in \mathcal{P}_n$, one introduces the mapping:
\begin{equation}\label{eq:iso_sympl_pauli}
 (1,0) \mapsto X, \, (0,1) \mapsto Z.
\end{equation}
We will use the symplectic vector notation $v=(x|z)$ to denote the $x$- and $z$-part in a $ v \in \Ftwotwon$. Eqn.~\eqref{eq:iso_sympl_pauli} can be straightforwardly extended to multiple qubits, i.e. to $\Ftwotwon$. Addition in $\Ftwotwon$ corresponds to multiplication in $\mathcal{P}_n$. For example,  multiplying $Y$ and $X$ corresponds to $(1|1)+(1|0)=(0|1)$. The symplectic product determines commutation of any $v$ and $w \in \Ftwotwon$, or rather the commutation of the corresponding Pauli operators. $v$ and $w$ commute if their symplectic product, as defined in Eqn.~\eqref{eq:symplecticprod}, is zero and anti-commute if their symplectic product is one. In a stabiliser code, in the absence of errors, it suffices to  use the symplectic space despite it being oblivious to phases. 

A \emph{stabiliser tableau} $T$ is a full-rank $r \times (2n)$ binary matrix where each row corresponds to the symplectic representation of each generator. The stabiliser group is abelian and thus any two rows $v,w$ fulfil $\langle v,w \rangle =0$. $\mathcal{N}(\Scurvy)$ corresponds to the kernel of $T\Omega$ and $\mathcal{N}(\Scurvy)/\Scurvy$ corresponds to $\ker (T\Omega)/\mathcal{R}(T)$, where we denote the row span of some matrix $M$ by $\mathcal{R}(M)$. $\mathcal{N}(\Scurvy)/\Scurvy$ is a symplectic vector space, a vector space with a symplectic form. A \emph{symplectic basis} of a symplectic vector space is a basis consisting of two sets $\{e_i\}$ and $\{f_i\}$, where $\langle e_i,e_j \rangle =0$ and $\langle f_i,f_j \rangle =0$ $\forall i,j$, such that $\langle e_i, f_i \rangle=1$ and $\langle e_i, f_j \rangle =0$ $\forall i \neq j$. A symplectic basis of a symplectic vector space always exists. (See Ref.~\cite{EslamiRad2024Symplectic} for details.)  

 Throughout this work we will often need to refer to the $i$-th component of a $P \in \Pn$. To do so, we will write $P_i$. Moreover, for a qubit index subset $I \subseteq Q$ we introduce the notation $P_I$ to refer to the \emph{restriction} of $P$ onto the qubits in $I$. Analogously, we introduce the notation $v[I]$, which we define as the restriction of $v \in \Ftwotwon$ on the indices in $\symplspace$ corresponding to the qubits in $I$. For example, if $I=\{0,1\}$ and $v=(0,0,1|1,0,0)$, then $v[I]=(0,0|1,0)$. 

The \emph{support} of any $P$ $\in \Pn$ is the subset of qubits on which $P$ acts non-trivially. For example, the support of $X_0X_1$ is $\{0,1\}$. 
 We will also introduce a short-form for the number of qubits on which two Paulis $P, Q$ $\in \Pn$ anti-commute \emph{qubit-wise}:
 \begin{equation}\label{eq:qubit_wise_anticomm}
     \alpha(P,Q)=\left | \{i \in (\supp(P) \cap \supp(Q))| \,  \{P_i,Q_i\}=0 \}\right|.
 \end{equation}

\subsection{Graph codes}\label{sec:graph_codes}
Here we discuss stabiliser codes defined with respect to graphs.
A stabiliser state is a state that is uniquely determined as the joint positive eigenstate of a stabiliser group, meaning the dimension of the joint positive eigenspace of the stabiliser group is one. A graph state is a stabiliser state induced by a graph $G=(V,E)$ with vertex set $V$ and edge set $E$. Given a graph, the stabilisers of the graph state are defined as follows: Qubits sit on the nodes. One can define a generating set for $\Scurvy$ by defining a stabiliser $s_v$ for each qubit $v \in V$:
\begin{equation}\label{eq:graph_state_generators}
    s_v=X_v \prod_{w \in \mathcal{N}(v)} Z_w,
\end{equation}
where $\mathcal{N}(v)$ denotes the neighbourhood of vertex $v$. In order to turn a graph state into a graph code, one usually entangles the graph state with an input qubit using CZ-gates, measures the input qubit in the $X$-basis and subsequently discards the input qubit, which results in a graph code encoding one qubit (see Ref.~\cite{Hwang_2016}). In Ref.~\cite{morley-shortLosstolerantTeleportationLarge2019}, the authors do not perform the measurement and the input qubit is still considered part of the code. Hence, given a graph the corresponding stabilisers are given by Eqn.~\eqref{eq:graph_state_generators} for all but the input node. $Z_I$, where $I$ denotes the input qubit, and $X_I \prod_{w \in \mathcal{N}(I)} Z_w$ are representatives for the two non-trivial generators of $\mathcal{N}(\Scurvy)/\Scurvy$. The fact that the single-qubit operator $Z_I$ is an element of $\Ncurvy(\Scurvy)-\Scurvy$ means that all logical information is destroyed if the input qubit is lost. In this work, in the case of a designated target qubit, we follow the same convention as in Ref.~\cite{morley-shortLosstolerantTeleportationLarge2019}, namely that the input and the target qubit are never lost. 

 We visualise the so-called `crazy graph' defining the crazy graph code from Ref.~\cite{morley-shortLosstolerantTeleportationLarge2019} in Fig.~\ref{fig:channels}. We will use the crazy graph to benchmark our algorithms against the ones in Ref.~\cite{morley-shortLosstolerantTeleportationLarge2019}. The crazy graph code is parametrised by a lattice parameter $l$, which is related to the qubit number $n$ by $n=l^2+2$.

\begin{figure}[t]
    \centering
 \begin{subfigure}[b]{0.45\textwidth}
        \centering
      \begin{tikzpicture}[xscale=1.50, yscale=1.10,
    nd/.style={draw, circle, inner sep=1.3pt, font=\scriptsize, fill=white, line width=0.7pt},
    onode/.style={draw=red!80!black, circle, inner sep=1.3pt, font=\scriptsize, fill=red!80!black, text=white, line width=0.7pt},
    ed/.style={gray!65, line width=0.9pt}]
    \draw[ed] (0.000,0.000) -- (1.000,1.500);
    \draw[ed] (0.000,0.000) -- (1.000,0.500);
    \draw[ed] (0.000,0.000) -- (1.000,-0.500);
    \draw[ed] (0.000,0.000) -- (1.000,-1.500);
    \draw[ed] (1.000,1.500) -- (1.000,0.500);
    \draw[ed] (1.000,0.500) -- (1.000,-0.500);
    \draw[ed] (1.000,-0.500) -- (1.000,-1.500);
    \draw[ed] (2.000,1.500) -- (2.000,0.500);
    \draw[ed] (2.000,0.500) -- (2.000,-0.500);
    \draw[ed] (2.000,-0.500) -- (2.000,-1.500);
    \draw[ed] (1.000,1.500) -- (2.000,1.500);
    \draw[ed] (1.000,0.500) -- (2.000,0.500);
    \draw[ed] (1.000,-0.500) -- (2.000,-0.500);
    \draw[ed] (1.000,-1.500) -- (2.000,-1.500);
    \draw[ed] (1.000,0.500) -- (2.000,1.500);
    \draw[ed] (1.000,-0.500) -- (2.000,0.500);
    \draw[ed] (1.000,-1.500) -- (2.000,-0.500);
    \draw[ed] (3.000,1.500) -- (3.000,0.500);
    \draw[ed] (3.000,0.500) -- (3.000,-0.500);
    \draw[ed] (3.000,-0.500) -- (3.000,-1.500);
    \draw[ed] (2.000,1.500) -- (3.000,1.500);
    \draw[ed] (2.000,0.500) -- (3.000,0.500);
    \draw[ed] (2.000,-0.500) -- (3.000,-0.500);
    \draw[ed] (2.000,-1.500) -- (3.000,-1.500);
    \draw[ed] (2.000,1.500) -- (3.000,0.500);
    \draw[ed] (2.000,0.500) -- (3.000,-0.500);
    \draw[ed] (2.000,-0.500) -- (3.000,-1.500);
    \draw[ed] (4.000,1.500) -- (4.000,0.500);
    \draw[ed] (4.000,0.500) -- (4.000,-0.500);
    \draw[ed] (4.000,-0.500) -- (4.000,-1.500);
    \draw[ed] (3.000,1.500) -- (4.000,1.500);
    \draw[ed] (3.000,0.500) -- (4.000,0.500);
    \draw[ed] (3.000,-0.500) -- (4.000,-0.500);
    \draw[ed] (3.000,-1.500) -- (4.000,-1.500);
    \draw[ed] (3.000,0.500) -- (4.000,1.500);
    \draw[ed] (3.000,-0.500) -- (4.000,0.500);
    \draw[ed] (3.000,-1.500) -- (4.000,-0.500);
    \draw[ed] (4.000,1.500) -- (5.000,0.000);
    \draw[ed] (4.000,0.500) -- (5.000,0.000);
    \draw[ed] (4.000,-0.500) -- (5.000,0.000);
    \draw[ed] (4.000,-1.500) -- (5.000,0.000);
    \node[nd] at (0.000,0.000) {0};
    \node[nd] at (1.000,1.500) {1};
    \node[nd] at (1.000,0.500) {2};
    \node[nd] at (1.000,-0.500) {3};
    \node[nd] at (1.000,-1.500) {4};
    \node[nd] at (2.000,1.500) {5};
    \node[nd] at (2.000,0.500) {6};
    \node[nd] at (2.000,-0.500) {7};
    \node[nd] at (2.000,-1.500) {8};
    \node[nd] at (3.000,1.500) {9};
    \node[nd] at (3.000,0.500) {10};
    \node[nd] at (3.000,-0.500) {11};
    \node[nd] at (3.000,-1.500) {12};
    \node[nd] at (4.000,1.500) {13};
    \node[nd] at (4.000,0.500) {14};
    \node[nd] at (4.000,-0.500) {15};
    \node[nd] at (4.000,-1.500) {16};
    \node[onode] at (5.000,0.000) {17};
  \end{tikzpicture}
      \caption{Triangular graph}
         \label{fig:first}
     \end{subfigure}
     \hfill
    \begin{subfigure}[b]{0.48\textwidth}
        \centering
        \begin{tikzpicture}[xscale=1.60, yscale=1.10,
    nd/.style={draw, circle, inner sep=1.3pt, font=\scriptsize, fill=white, line width=0.7pt},
    onode/.style={draw=red!80!black, circle, inner sep=1.3pt, font=\scriptsize, fill=red!80!black, text=white, line width=0.7pt},
    ed/.style={gray!65, line width=0.9pt}]
    \draw[ed] (0.000,0.000) -- (1.000,1.500);
    \draw[ed] (0.000,0.000) -- (1.000,0.500);
    \draw[ed] (0.000,0.000) -- (1.000,-0.500);
    \draw[ed] (0.000,0.000) -- (1.000,-1.500);
    \draw[ed] (1.000,1.500) -- (2.000,1.500);
    \draw[ed] (1.000,0.500) -- (2.000,1.500);
    \draw[ed] (1.000,-0.500) -- (2.000,1.500);
    \draw[ed] (1.000,-1.500) -- (2.000,1.500);
    \draw[ed] (1.000,1.500) -- (2.000,0.500);
    \draw[ed] (1.000,0.500) -- (2.000,0.500);
    \draw[ed] (1.000,-0.500) -- (2.000,0.500);
    \draw[ed] (1.000,-1.500) -- (2.000,0.500);
    \draw[ed] (1.000,1.500) -- (2.000,-0.500);
    \draw[ed] (1.000,0.500) -- (2.000,-0.500);
    \draw[ed] (1.000,-0.500) -- (2.000,-0.500);
    \draw[ed] (1.000,-1.500) -- (2.000,-0.500);
    \draw[ed] (1.000,1.500) -- (2.000,-1.500);
    \draw[ed] (1.000,0.500) -- (2.000,-1.500);
    \draw[ed] (1.000,-0.500) -- (2.000,-1.500);
    \draw[ed] (1.000,-1.500) -- (2.000,-1.500);
    \draw[ed] (2.000,1.500) -- (3.000,1.500);
    \draw[ed] (2.000,0.500) -- (3.000,1.500);
    \draw[ed] (2.000,-0.500) -- (3.000,1.500);
    \draw[ed] (2.000,-1.500) -- (3.000,1.500);
    \draw[ed] (2.000,1.500) -- (3.000,0.500);
    \draw[ed] (2.000,0.500) -- (3.000,0.500);
    \draw[ed] (2.000,-0.500) -- (3.000,0.500);
    \draw[ed] (2.000,-1.500) -- (3.000,0.500);
    \draw[ed] (2.000,1.500) -- (3.000,-0.500);
    \draw[ed] (2.000,0.500) -- (3.000,-0.500);
    \draw[ed] (2.000,-0.500) -- (3.000,-0.500);
    \draw[ed] (2.000,-1.500) -- (3.000,-0.500);
    \draw[ed] (2.000,1.500) -- (3.000,-1.500);
    \draw[ed] (2.000,0.500) -- (3.000,-1.500);
    \draw[ed] (2.000,-0.500) -- (3.000,-1.500);
    \draw[ed] (2.000,-1.500) -- (3.000,-1.500);
    \draw[ed] (3.000,1.500) -- (4.000,1.500);
    \draw[ed] (3.000,0.500) -- (4.000,1.500);
    \draw[ed] (3.000,-0.500) -- (4.000,1.500);
    \draw[ed] (3.000,-1.500) -- (4.000,1.500);
    \draw[ed] (3.000,1.500) -- (4.000,0.500);
    \draw[ed] (3.000,0.500) -- (4.000,0.500);
    \draw[ed] (3.000,-0.500) -- (4.000,0.500);
    \draw[ed] (3.000,-1.500) -- (4.000,0.500);
    \draw[ed] (3.000,1.500) -- (4.000,-0.500);
    \draw[ed] (3.000,0.500) -- (4.000,-0.500);
    \draw[ed] (3.000,-0.500) -- (4.000,-0.500);
    \draw[ed] (3.000,-1.500) -- (4.000,-0.500);
    \draw[ed] (3.000,1.500) -- (4.000,-1.500);
    \draw[ed] (3.000,0.500) -- (4.000,-1.500);
    \draw[ed] (3.000,-0.500) -- (4.000,-1.500);
    \draw[ed] (3.000,-1.500) -- (4.000,-1.500);
    \draw[ed] (4.000,1.500) -- (5.000,0.000);
    \draw[ed] (4.000,0.500) -- (5.000,0.000);
    \draw[ed] (4.000,-0.500) -- (5.000,0.000);
    \draw[ed] (4.000,-1.500) -- (5.000,0.000);
    \node[nd] at (0.000,0.000) {0};
    \node[nd] at (1.000,1.500) {1};
    \node[nd] at (1.000,0.500) {2};
    \node[nd] at (1.000,-0.500) {3};
    \node[nd] at (1.000,-1.500) {4};
    \node[nd] at (2.000,1.500) {5};
    \node[nd] at (2.000,0.500) {6};
    \node[nd] at (2.000,-0.500) {7};
    \node[nd] at (2.000,-1.500) {8};
    \node[nd] at (3.000,1.500) {9};
    \node[nd] at (3.000,0.500) {10};
    \node[nd] at (3.000,-0.500) {11};
    \node[nd] at (3.000,-1.500) {12};
    \node[nd] at (4.000,1.500) {13};
    \node[nd] at (4.000,0.500) {14};
    \node[nd] at (4.000,-0.500) {15};
    \node[nd] at (4.000,-1.500) {16};
    \node[onode] at (5.000,0.000) {17};
  \end{tikzpicture}
        \caption{Crazy graph}
        \label{fig:second}
    \end{subfigure}
\caption{The triangular graph and the crazy graph ($4\times4$, i.e.\ 18 qubits) from Ref.~\cite{morley-shortLosstolerantTeleportationLarge2019}. The input qubit is qubit 0 and the red node is the target qubit.}
\label{fig:channels}
\end{figure}
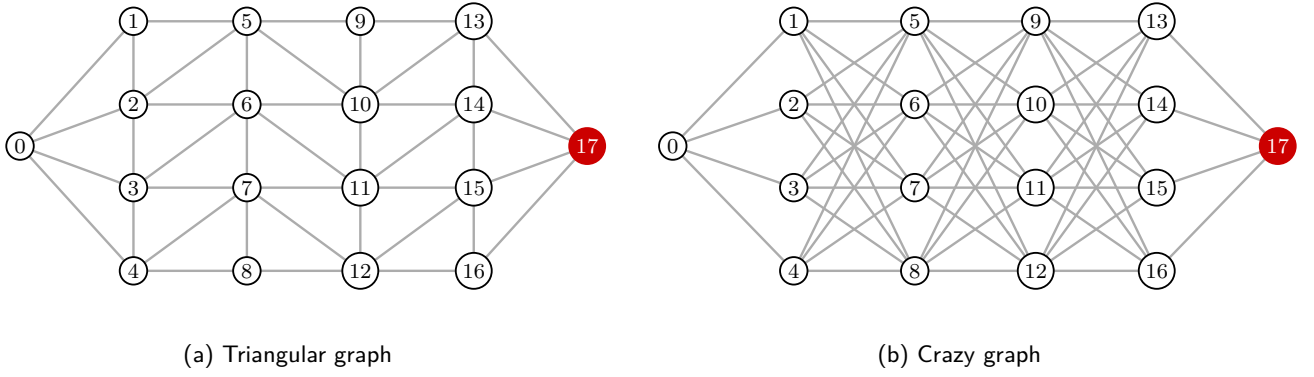

Technically, we discuss families of graph codes. The graph codes shown in Fig.~\ref{fig:channels} are $(4 \times 4)$, i.e. $l=4$. One can straightforwardly extend them to larger $(l \times l)$ codes by adding columns or rows of nodes and following the same pattern for the edges. 
 
\subsection{Planar surface code}\label{sec:surface_code}

In this section, we briefly introduce the planar surface code on a square lattice with $\Lambda\times(\Lambda-1)$ vertices. We will define its stabilisers in a slightly informal manner and will refrain from mathematical rigour beyond what is necessary to understand the proof for the $g$-SPF threshold. For a more detailed introduction to surface codes see~\cite{dennisTopologicalQuantumMemory2002,fowlerSurfaceCodesPractical2012,lidarQuantumErrorCorrection2013}. The most important take-away from the paragraphs below is that the representatives of logical $X$ operators are edges along a path from left to right in a lattice that is $\Lambda$ edges wide and $\Lambda-1$ edges high, and the representatives of logical $Z$ operators are edges along a path from top to bottom in its dual lattice, which is $\Lambda-1$ edges wide and $\Lambda$ edges high, which will enable the application of standard percolation theory on a square lattice later. 

The surface code is a topological stabiliser code defined on a lattice. There are multiple variations of the surface code, but here we consider the planar surface code (with open boundary conditions) defined on a square lattice with vertex set $V$, edge set $E$ and face set $F$. Qubits are attached to each edge $e \in E$. The surface code is a CSS code~\cite{calderbankGoodQuantumErrorcorrecting1996,steaneMultipleparticleInterferenceQuantum1996,nielsenQuantumComputationQuantum2010} and has $X$-type stabilisers (plaquettes) that are defined as the product of $X$-operators on edges along a face of the square lattice and $Z$-type stabilisers (stars) that are products of $Z$-operators neighbouring a vertex $v$. To be precise, an $X$-type stabiliser is defined as $X_f:=\prod_{e \in N_E(f)} X_e$, where $N_E(f)$ denotes the set of edges bounding the face $f \in F$, and a $Z$-type stabiliser is defined as $Z_v:=\prod_{e \in N_E(v)} Z_e$, where $N_E(v)$ denotes the set of edges incident to the vertex $v \in V$. See Fig.~\ref{fig:sc-stab} for a visualisation. 

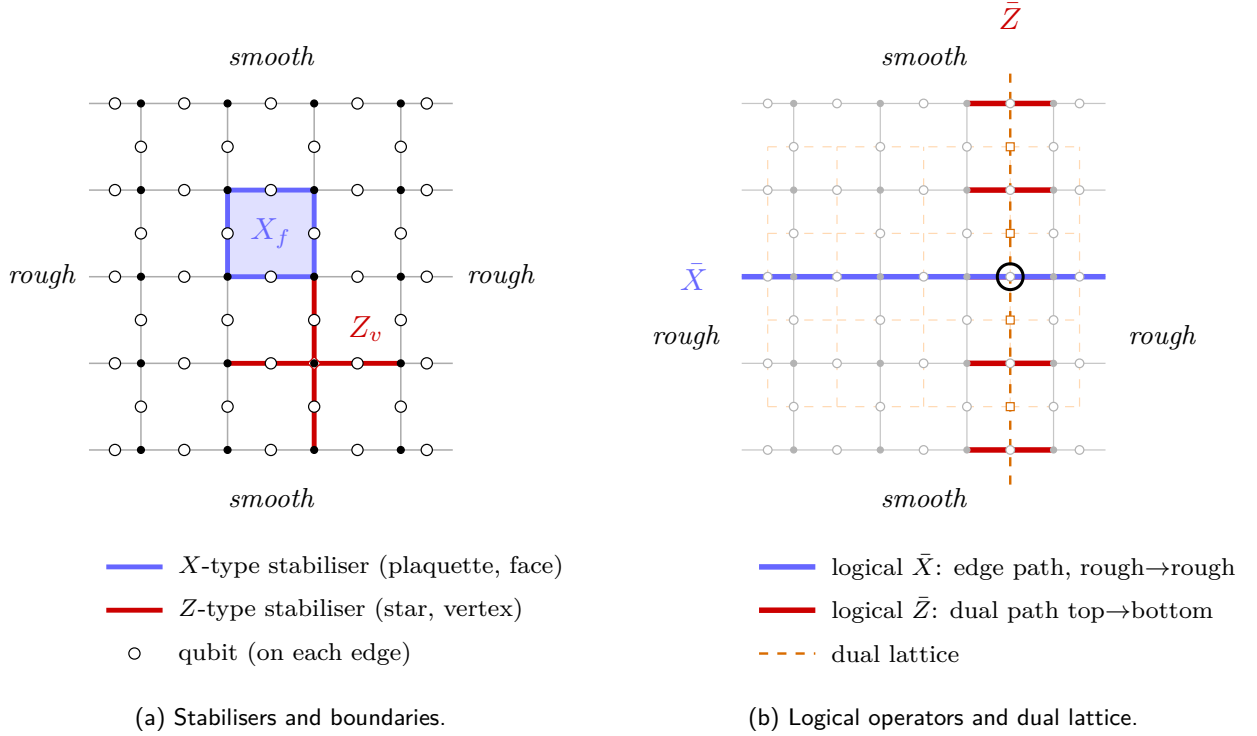
\begin{figure}[htbp]\centering
\begin{subfigure}[b]{0.48\textwidth}\centering
\resizebox{\linewidth}{!}{%
\begin{tikzpicture}[scale=1.0,
 latt/.style={gray!65,line width=0.5pt},
 xstab/.style={blue!60!white,line width=1.6pt},
 zstab/.style={red!80!black,line width=1.6pt},
 q/.style={draw,circle,fill=white,inner sep=1.3pt},
 vtx/.style={draw,circle,fill=black,inner sep=0.8pt},
 bl/.style={font=\footnotesize\itshape}]
\useasboundingbox (-1.75,-2.75) rectangle (4.95,5.25);
\fill[blue!12] (1,2) rectangle (2,3);
\foreach \y in {0,...,4}{\foreach \x in {0,...,2}{\draw[latt] (\x,\y)--(\x+1,\y);}}
\foreach \x in {0,...,3}{\foreach \y in {0,...,3}{\draw[latt] (\x,\y)--(\x,\y+1);}}
\foreach \y in {0,...,4}{\draw[latt] (-0.6,\y)--(0,\y); \draw[latt] (3,\y)--(3.6,\y);}
\draw[xstab] (1,2)--(2,2)--(2,3)--(1,3)--cycle;
\draw[zstab] (1,1)--(3,1); \draw[zstab] (2,0)--(2,2);
\node[red!80!black,fill=white,draw,diamond,inner sep=1.1pt] at (2,1) {};
\foreach \y in {0,...,4}{\foreach \x in {0.5,1.5,2.5,-0.3,3.3}{\node[q] at (\x,\y){};}}
\foreach \x in {0,...,3}{\foreach \y in {0.5,1.5,2.5,3.5}{\node[q] at (\x,\y){};}}
\foreach \x in {0,...,3}{\foreach \y in {0,...,4}{\node[vtx] at (\x,\y){};}}
\node[blue!60!white,font=\small] at (1.5,2.5) {$X_f$};
\node[red!80!black,font=\small] at (2.6,1.4) {$Z_v$};
\node[bl] at (-1.15,2){rough};
\node[bl] at (4.15,2){rough};
\node[bl] at (1.5,4.55){smooth};
\node[bl] at (1.5,-0.55){smooth};
\draw[xstab] (-0.4,-1.35)--(0.25,-1.35); \node[anchor=west,font=\scriptsize] at (0.3,-1.35){$X$-type stabiliser (plaquette, face)};
\draw[zstab] (-0.4,-1.85)--(0.25,-1.85); \node[anchor=west,font=\scriptsize] at (0.3,-1.85){$Z$-type stabiliser (star, vertex)};
\node[q] at (-0.075,-2.35){}; \node[anchor=west,font=\scriptsize] at (0.3,-2.35){qubit (on each edge)};
\end{tikzpicture}
}
\caption{Stabilisers and boundaries.}\label{fig:sc-stab}
\end{subfigure}\hfill
\begin{subfigure}[b]{0.48\textwidth}\centering
\resizebox{\linewidth}{!}{%
\begin{tikzpicture}[scale=1.0,
 latt/.style={gray!45,line width=0.4pt},
 q/.style={draw=gray!60,circle,fill=white,inner sep=1.0pt},
 vtx/.style={draw=gray!60,circle,fill=gray!60,inner sep=0.7pt},
 dual/.style={orange!85!black,dashed,line width=0.8pt},
 dualf/.style={orange!30,dashed,line width=0.4pt},
 dnode/.style={draw=orange!85!black,rectangle,fill=white,inner sep=1.3pt},
 xlog/.style={blue!60!white,line width=1.8pt},
 zlog/.style={red!80!black,line width=1.8pt},
 bl/.style={font=\footnotesize\itshape}]
\useasboundingbox (-1.75,-2.75) rectangle (4.95,5.25);
\foreach \y in {0,...,4}{\foreach \x in {0,...,2}{\draw[latt] (\x,\y)--(\x+1,\y);}}
\foreach \x in {0,...,3}{\foreach \y in {0,...,3}{\draw[latt] (\x,\y)--(\x,\y+1);}}
\foreach \y in {0,...,4}{\draw[latt] (-0.6,\y)--(0,\y); \draw[latt] (3,\y)--(3.6,\y);}
\foreach \y in {0.5,1.5,2.5,3.5}{
  \draw[dualf] (-0.3,\y)--(0.5,\y);
  \draw[dualf] (0.5,\y)--(1.5,\y)--(2.5,\y);
  \draw[dualf] (2.5,\y)--(3.3,\y);}
\foreach \x in {-0.3,0.5,1.5,2.5,3.3}{\draw[dualf] (\x,0.5)--(\x,3.5);}
\draw[dual] (2.5,-0.4) -- (2.5,4.4);
\foreach \y in {0,...,4}{\draw[zlog] (2,\y)--(3,\y);}
\draw[xlog] (-0.6,2)--(3.6,2);
\foreach \y in {0,...,4}{\foreach \x in {0.5,1.5,2.5,-0.3,3.3}{\node[q] at (\x,\y){};}}
\foreach \x in {0,...,3}{\foreach \y in {0.5,1.5,2.5,3.5}{\node[q] at (\x,\y){};}}
\foreach \x in {0,...,3}{\foreach \y in {0,...,4}{\node[vtx] at (\x,\y){};}}
\foreach \y in {0.5,1.5,2.5,3.5}{\node[dnode] at (2.5,\y){};}
\node[draw=black,circle,line width=1pt,inner sep=3pt] at (2.5,2) {};
\node[blue!60!white,font=\small] at (-1.15,2){$\bar X$};
\node[red!80!black,font=\small] at (2.5,5.02){$\bar Z$};
\node[bl] at (-1.25,1.3){rough};
\node[bl] at (4.25,1.3){rough};
\node[bl] at (1.5,4.55){smooth};
\node[bl] at (1.5,-0.55){smooth};
\draw[xlog] (-0.4,-1.35)--(0.25,-1.35); \node[anchor=west,font=\scriptsize] at (0.3,-1.35){logical $\bar X$: edge path, rough$\to$rough};
\draw[zlog] (-0.4,-1.85)--(0.25,-1.85); \node[anchor=west,font=\scriptsize] at (0.3,-1.85){logical $\bar Z$: dual path top$\to$bottom};
\draw[dual] (-0.4,-2.35)--(0.25,-2.35); \node[anchor=west,font=\scriptsize] at (0.3,-2.35){dual lattice};
\end{tikzpicture}
}
\caption{Logical operators and dual lattice.}\label{fig:sc-log}
\end{subfigure}
\caption{The planar surface code on a square lattice with $\Lambda\times(\Lambda-1)$ vertices (here $\Lambda=5$) and qubits on edges. (a) $X$-type stabilisers are plaquettes (products of $X$ on the edges around a face) and $Z$-type stabilisers are stars (products of $Z$ on the edges around a vertex); the left/right boundaries are rough and the top/bottom boundaries are smooth. (b) Any path of edges between two rough boundaries is a logical $X$ representative (left to right); any path between the two rough boundaries of the \emph{dual} lattice is a logical $Z$ representative (top to bottom), whose corresponding direct-lattice edges are shown in red.}\label{fig:surface}
\end{figure}

So far we have not considered any boundary conditions. We only defined $X$- and $Z$-type stabilisers in the bulk.
We consider a surface code with so-called `smooth boundaries'  at the top and bottom and `rough boundaries' at the left and right side. It can be constructed by first starting out with an infinite lattice and then considering a sublattice with $\Lambda$ rows and $\Lambda-1$ columns of vertices as the lattice on which we define the code. The smooth boundary is the edge set along the top or bottom of the lattice. All faces above or below the boundaries and all edges that only border on such a face are removed. The resulting $Z$-type stabilisers now become products of single-qubit $Z$ operators on the three edges in the neighbourhood of each vertex on the smooth boundary. The rough boundary is a result of removing vertices along the left and right side of the lattice. This induces the deletion of any edge that has two deleted vertices in its neighbourhood. The resulting $X$-type stabilisers along a rough boundary are still products of edges along one face, only including the non-deleted edges, meaning that along the rough boundary they are plaquettes that are open on the boundary. For simplification one can introduce the dual lattice, where faces and vertices are exchanged and edges connect the new vertices. Then, smooth and rough boundaries and plaquette and star stabilisers are switched. We added rough and smooth boundaries and a sketch of the dual lattice in Fig.~\ref{fig:surface}. The planar surface code has $\Lambda^2$ qubits on horizontal edges ($\Lambda$ per row, including the two edges dangling from the rough boundaries) and $(\Lambda-1)^2$ qubits on vertical edges, i.e.\ $n=\Lambda^2+(\Lambda-1)^2$ in total. Measured in edges, the lattice is thus $\Lambda$ wide and $\Lambda-1$ high; we refer to it as the $\Lambda\times(\Lambda-1)$ lattice. Its dual lattice has the $(\Lambda-1)\times\Lambda$ faces (including the open plaquettes along the rough boundaries) as vertices and is $\Lambda-1$ edges wide and $\Lambda$ edges high; we refer to it as the $(\Lambda-1)\times\Lambda$ dual lattice.

The planar surface code encodes one logical qubit. Any product of single-qubit $X$ operators on edges along a path from a rough boundary to another one, i.e. the edges along a path from left to right, is a representative of logical $X$. Any product of single-qubit $Z$ operators on edges along a path in the \emph{dual} lattice from a rough boundary to another rough boundary is a representative of logical $Z$. This is visualised in Fig.~\ref{fig:sc-log}. 
To obtain the representatives of $Z$ in the actual lattice (and not its dual), one maps each edge in the path on the dual lattice to its corresponding edge in the lattice. We show such a mapping in Fig.~\ref{fig:sc-log}.  
However, the dual lattice serves as a more convenient way to represent $Z$ representatives, which we will also use in the rest of this work. In summary, the representatives of $\bar X$ are the paths from left to right in the $\Lambda\times(\Lambda-1)$ lattice and the representatives of $\bar Z$ are the paths from top to bottom in the $(\Lambda-1)\times\Lambda$ dual lattice. Both lattices are a square lattice of $\Lambda$ edges by $(\Lambda-1)$ edges, one rotated by $90^{\circ}$ with respect to the other, and the minimum-weight representatives of $\bar X$ and $\bar Z$ both have weight $\Lambda$.

 \section{Stabiliser path finding}\label{sec:formalisation}

In this section we discuss the task of \emph{stabiliser path finding} (SPF), as first defined in Ref.~\cite{morley-shortLosstolerantTeleportationLarge2019}. We first give an overview of the idea, describe the method to solve SPF described in Ref.~\cite{morley-shortLosstolerantTeleportationLarge2019} and then formalise SPF and our relaxed version of it, where the goal is to localise information onto at most $g$ qubits. 

\subsection{SPF in previous literature}\label{sec:SPFprevious}

 In SPF the goal is to find two representatives $X$ and $Z$ that only anti-commute on a target qubit $t$ \emph{and} that both have no support on any lost qubits. We will call this pair an `SPF pair' or a `measurement pattern'. Measuring $X$ and $Z$ qubit-wise on all but the target qubit creates representatives of logical $X$ and $Z$ that correspond to single-qubit Paulis on the target qubit. Hence, information has been localised onto the target qubit.

The method in Ref.~\cite{morley-shortLosstolerantTeleportationLarge2019,sammorley_short_spf_2019} consists of listing all representatives of $X$ and $Z$ that can yield a valid measurement pattern. The main observation is that not every representative is suitable for a pattern: those that can be written as a product of some other representative $X'$ and a stabiliser $s$, where the supports of $X'$ and $s$ do not overlap, need not be considered. Thus, such operators, dubbed `trivial logical operators', are first filtered out. The authors do so by first creating a list of all non-trivial stabilisers, which are defined analogously, and then multiplying a representative of each logical class by non-trivial stabilisers and checking whether the result is non-trivial. They then create a list of all combinations of $X$ and $Z$ that anti-commute qubit-wise only on the target qubit. Then, for each set of lost qubits, one can search in the list for valid measurement patterns. The subsequent list search is fast; however, even after filtering out trivial logical operators, the creation of the list scales exponentially with the number of qubits, rendering it infeasible for large codes. Deciding whether a given representative is trivial or not also constitutes a complicated problem.
  
Moreover, the suggested teleportation protocol in Ref.~\cite{morley-shortLosstolerantTeleportationLarge2019} using SPF makes two assumptions: Firstly, the authors consider the designated input qubit as part of the graph code, in contrast to the usual understanding of graph codes, where the input qubit is entangled with a graph state and subsequently discarded (see for example Ref.~\cite{Hwang_2016}). Therefore, the code still contains a single-qubit logical operator supported only on the input qubit. If this input qubit is lost, all logical information is lost. The teleportation procedure from Ref.~\cite{morley-shortLosstolerantTeleportationLarge2019} thus requires a physical model where the input qubit may be part of more robust hardware, such as superconducting qubits, and is not sent to the receiver; or, if the input qubit, too, is encoded into a photon and sent to the receiver, presupposes that the input qubit is never lost. As a consequence, one either needs to be able to preserve long-range entanglement or think of some error mitigation/correction schemes that protect individual qubits from loss. Secondly, the target qubit (or output qubit, as they call it) must not be lost either. 

It should be noted that Ref.~\cite{PhysRevA.110.052617} discusses the localisation of logical information onto one qubit in a stabiliser code via local operations and classical communication (LOCC), where the authors also prove that a protocol to localise information to a specified target can be found in polynomial time. The main difference from the setting in Ref.~\cite{morley-shortLosstolerantTeleportationLarge2019} is that in the latter the task is to search for a localisation protocol \emph{after} qubit loss, for which the arguments in Ref.~\cite{PhysRevA.110.052617} do not apply. We suspect that the hardness of the localisation task stems from qubit loss.

 \subsection{Formalisation}
 
In this section, we formalise Target-SPF and our generalisation $g$-SPF into an optimisation problem. To be clear, we call the problem the authors of Ref.~\cite{morley-shortLosstolerantTeleportationLarge2019} address `Target-SPF' and our generalisation `$g$-SPF'. In both cases the goal is to find a pair of logical representatives $X \in [X]$ and $Z \in [Z]$. In Ref.~\cite{morley-shortLosstolerantTeleportationLarge2019}, the authors establish that in order for the localisation onto a target qubit $t$ to be successful, $X$ and $Z$ must fulfil the so-called \emph{stabiliser path finding} conditions, which we summarise as the \emph{Lost-Qubits condition} and the \emph{Target condition}: 
\begin{definition}[Lost-Qubits condition]\label{def:Lemptyset}
    Logical operators $X \in [X]$, $Z \in [Z]$ fulfil the Lost-Qubits condition if:
     \begin{equation} (\supp(X) \cup \supp(Z)) \cap L=\emptyset, \end{equation} 
     where $L \subseteq Q$ is the set of lost qubits.
\end{definition}
Recall that $Q$ is the set of all qubits.
\begin{definition}[Target condition]\label{def:tanti}
Operators $X \in [X]$ and $Z \in [Z]$ fulfil the Target condition if:
     \begin{equation} \{X_t,Z_t\}=0, \quad [X_i,Z_i]=0 \quad \forall \, i \neq t. \end{equation} 
\end{definition}
These two conditions specify that valid logical operators $X$ and $Z$ must satisfy qubit-wise anti-commutation only on the target qubit $t$ and commutation otherwise and they should have no support on the qubits that have been lost.

After creating the list of all measurement patterns that fulfil the Target condition, the authors of Ref.~\cite{morley-shortLosstolerantTeleportationLarge2019} search for SPF pairs that fulfil the Lost-Qubits condition and that minimise the total support of $X$ and $Z$, i.e. $|\supp(X) \cup \supp(Z)|$. An SPF pair with minimal total support is preferable to make the localisation procedure more robust in the case of imperfect measurements. In this work, we consider the minimisation of the support of logical operators as our main goal. Our deterministic algorithm D-LoFi finds the exact minimum, whereas our heuristic algorithm H-LoFi prefers solutions $(X,Z)$ with small support due to its cost function. Note, however, that to merely find a measurement pattern, minimisation is not strictly necessary.

We consider two types of SPF: The first, dubbed `Target-SPF', is the same SPF task as in Ref.~\cite{morley-shortLosstolerantTeleportationLarge2019}. We summarise it below.
\begin{iobox}[prob:tSPF]{Target-SPF}
\tcbsubtitle{Input}
Stabiliser group $\Scurvy$, target qubit index $t$, set of lost qubits $L \subseteq Q$.

\tcbsubtitle{Output}
An $X \in [X]$, $Z \in [Z]$ such that:

\begin{itemize}
    \item The Lost-Qubits condition
    \item The Target condition
    \item optional: $|\supp(X) \cup \supp(Z)|$ minimal
\end{itemize}
\end{iobox}
When computing thresholds for Target-SPF, the authors of Ref.~\cite{morley-shortLosstolerantTeleportationLarge2019} assume that the target qubit is never lost. This may not be a realistic assumption. Moreover, there may be other potential target qubits for which an $SPF$-pair may be found, even if there does not exist one for qubit $t$. Localisation of information may thus be hindered by the arbitrary assignment of a target. We therefore introduce the second type of SPF-problem, which we call `$g$-SPF', where we relax the Target condition (Definition~\ref{def:tanti}) to the $g$-condition:
\begin{definition}[$g$-condition]\label{def:gcondition}
    Operators $X \in [X]$, $Z \in [Z]$ fulfil the $g$-condition if:
     \begin{equation} \alpha(X,Z) \leq g. \end{equation}
\end{definition}
Here $\alpha(X,Z)$ is defined as the number of qubits on which $X$ and $Z$ anti-commute, see Eqn.~\eqref{eq:qubit_wise_anticomm}. The $g$-SPF problem can thus be defined as: 
\begin{iobox}[prob:gSPF]{$g$-SPF}
\tcbsubtitle{Input}
Stabiliser group $\Scurvy$, integer $g$, set of lost qubits $L \subseteq Q$.
 
\tcbsubtitle{Output}
An $X \in [X]$, $Z \in [Z]$ such that $X$ and $Z$ fulfil the following conditions :

\begin{itemize}
    \item The Lost-Qubits condition
    \item The $g$-condition
    \item optional: $|\supp(X) \cup \supp(Z)|$ minimal
\end{itemize}
\end{iobox}
Note that if $g$ is one, $g$-SPF becomes Target-SPF with an arbitrary target qubit.  

\section{Localisation threshold for the planar surface code on a square lattice}\label{sec:threshold_surfacecode}
 
Here, we show that the planar surface code on a square lattice (defined in Section~\ref{sec:surface_code}) has a threshold for localisation via $g$-SPF, where $g$ is a constant. We will first clarify some notation and definitions and then state the formal version of Theorem~\ref{thm:threshold_surface_informal}. The rest of this section is then structured as follows: In Section~\ref{sec:relation_percolation} we will introduce basic percolation theory and how it relates to the $g$-SPF problem on the planar surface code. Then, in Section~\ref{sec:disjointness}, we will introduce the disjointness of stabiliser codes and derive a lemma that relates it to the intersection size of the $(X,Z)$ logical pairs in the surface code. We will combine both in Section~\ref{sec:proof_spfthreshold_surface}, where we prove Theorem~\ref{thm:threshold_surface_informal}.

Recall that we write $\{\Scurvy_n\}$ for a family of codes indexed by the code size $n$ (the number of qubits). Qubit losses are i.i.d. with respect to a qubit loss probability $p$.  We will write $Q$ for the set of all $n$ physical qubits, i.e. $|Q|=n$. We will now define what it means for a family $\{\Scurvy_n\}$ to exhibit a $g$-SPF threshold.  Given a loss probability, a set of lost qubits $L \subseteq Q$ can be sampled according to the i.i.d. loss probability distribution $P_p:  \mathbf{P}(Q) \rightarrow [0,1]$, where we write $\mathbf{P}(S)$ for the power set of a set $S$:
 \begin{equation}\label{eq:loss_prob}
     P_p(L)= p^{|L|}(1-p)^{n-|L|}.
 \end{equation}
 It is helpful to be precise about the terms we use. Technically speaking, we wish to define the probability of the event that there exists a $g$-SPF pair. Here, an event is a subset of loss configurations. Thus, the event that there exists a $g$-SPF pair is defined as:
\begin{equation}\label{eq:loss_config_SPFpairexists}
L_{\mathrm{success}}(\leq g):=\{L \in \mathbf{P}(Q)\big | g\text{-SPF pair exists for }L\}.
\end{equation}
The probability of  $L_{\mathrm{success}}(\leq g)$ is determined by:
\begin{equation}\label{eq:PgSPF}
    P_p( L_{\mathrm{success}}(\leq g))=\sum_{L \in L_{\mathrm{success}}(\leq g)}P_p(L),
\end{equation}
where $P_p(L)$ is given by Eqn.~\eqref{eq:loss_prob}. 

We say that a family of codes has an $\mathcal{O}(1)$-SPF threshold if for sufficiently large $n$, a $g$-SPF pair exists with high likelihood, where $g$ is a constant: 

\begin{definition}[Threshold for $\mathcal{O}(1)$-SPF] \label{def:threshold}
A family of codes $\{\Scurvy_n\}$ has an $\order{1}$-SPF threshold $p'$ if for any qubit loss probability $p<p'$ there exists a constant $g$ such that the following holds for the $g$-SPF success probability, as defined in Eqn.~\eqref{eq:PgSPF}: 
\begin{equation}
\lim_{n \rightarrow \infty} P_p( L_{\mathrm{success}}(\leq g)) =1,
\end{equation}
where $ L_{\mathrm{success}}(\leq g)$ is given by Eqn.~\eqref{eq:loss_config_SPFpairexists}.
\end{definition}
Note that $g$ may be different for each $p <p'$. The important point is that $g$ remains constant with respect to $n$. 

We will show that the $\mathcal{O}(1)$-SPF threshold for the planar surface code is at least $p'=1/2$.  
\begin{restatedtheorem}[$\order{1}$-SPF threshold for the planar surface code]\label{thm:threshold_surface_informal}
Let $\{\Scurvy_n\}$ be a family of planar surface codes. Consider the i.i.d. qubit erasure channel where each qubit is lost with the individual loss probability $p$. Then, for all $p<1/2$, there exists a strictly positive constant $g=g(p)$ such that: 
\begin{equation}\label{eq:limit_success_rate}
   \lim_{n \to \infty} P_p(L_{\mathrm{success}}(\leq g))=1,
\end{equation}
where $P_p(L_{\mathrm{success}}(\leq g))$ is defined in Eqn.~\eqref{eq:PgSPF}. Moreover, for all $p>1/2$ and any constant $g$:
\begin{equation}\label{eq:limit_success_rate_up}
   \lim_{n \to \infty} P_p(L_{\mathrm{success}}(\leq g))=0,
\end{equation}
meaning $p=1/2$ is the exact threshold.
\end{restatedtheorem}
Note that the event $L_{\mathrm{success}}(\leq g)$ and hence its probability $P_p(L_{\mathrm{success}}(\leq g))$ depend on the code $\Scurvy_n$. 
The standard erasure threshold of the planar surface code is $p_e=1/2$, which has been obtained using percolation theory in Ref.~\cite{staceThresholdsTopologicalCodes2009}.
It is defined with respect to the mere existence of logical information, which is not the same as the existence of a $g$-SPF pair. Hence one cannot directly apply the results from Ref.~\cite{staceThresholdsTopologicalCodes2009} to prove Theorem~\ref{thm:threshold_surface_informal}. Nevertheless, one can prove Theorem~\ref{thm:threshold_surface_informal} by combining known statements about bond percolation in a two-dimensional box with the notion of sets of \emph{disjoint} logical operators in stabiliser codes. In fact, we derive a lemma similar to the scrubbing lemma in Ref.~\cite{jochym-oconnorDisjointnessStabilizerCodes2018}, which relates the so-called disjointness of a stabiliser code to the minimum intersection size of the logical operator pairs $(X,Z)$. We will show, using percolation theory, that below the threshold $p'=1/2$, for a large enough $n$, there exist sets $X_R \subseteq [X]$ and $Z_R \subseteq [Z]$ of $\Theta(\sqrt{n})$ disjoint representatives for $X$ and $Z$. Then, using the modified scrubbing lemma, we will show that this implies that there exists a constant $g$ such that, for large enough $n$, the $g$-SPF success rate asymptotically approaches 1.

\subsection{Relation to percolation theory}\label{sec:relation_percolation}

Percolation theory studies the emergence of connected regions in networks when edges or nodes are added or removed. Here we will only discuss percolation in a two-dimensional rectangular lattice, i.e. bond percolation in two dimensions. We can view the lattice as a graph $G=(V,E)$. A fundamental percolation model is the following: Each edge in the lattice is \emph{open} or \emph{closed} following the i.i.d. probability distribution $\{1-p,p\}$. Let $\Omega:=\{0,1\}^E$. We summarise which edges are closed and open by denoting a closed edge by 0 and an open edge by 1. Let $\omega \in \Omega$ denote a specific edge configuration. 
We call any $\mathcal{A} \subseteq \Omega$ an \emph{event}. Note that $\omega$ is a characteristic vector encoding $Q-L$ for a loss configuration $L \subseteq Q$, as we defined it in Section~\ref{sec:formalisation}.

We say there exists an open path from the left to the right side of the lattice if there exists a connected path of edges where the first edge is attached to a vertex on the very left side of the lattice and the last edge of the path is attached to one on the right side. A fundamental question in percolation theory is the question for which probabilities $p$ there exists an open path, i.e. a connected set of edges, from the left to the right side of a two-dimensional square lattice. (Or equivalently from top to bottom.) That is, one asks whether there exists a $p'$ such that for all $p< p'$:
\begin{equation}\label{eq:percolation_threshold}
   \lim_{\Lambda \to \infty} P_{p}(\mathcal{A})=1,
\end{equation}
where $P_{p}$ is the probability of the event $\mathcal{A}$ for a specific $p$ and $\mathcal{A}$ is given by:
\begin{equation}
    \mathcal{A}=\{\omega \in \Omega\;|\;\omega\text{ contains an open path from left to right}\}.
\end{equation}
In Ref.~\cite{kestenCriticalProbabilityBond1980}, H. Kesten showed that the threshold $p'$ in Eqn.~\eqref{eq:percolation_threshold} is $p'=\frac{1}{2}$. The bond percolation model on a square lattice is in one-to-one correspondence with the planar surface code under qubit loss, where each qubit is lost with probability $p$. The existence of a surviving logical representative directly corresponds to the existence of an open path from one boundary to the opposite one, as pointed out in Ref.~\cite{staceThresholdsTopologicalCodes2009}. Therefore, the standard erasure threshold is $1/2$. 

Here, we care about the slightly different event:
\begin{equation}\label{eq:event_many_paths}
    \mathcal{A}=\{\omega \in \Omega\;|\;\omega  \text{ contains many edge-disjoint paths from left to right}\}.
\end{equation}
The wording `many edge-disjoint paths' is purposefully kept vague for comprehension purposes. We will specify it later. 
In particular, we wish to prove that for a specific $p'>0$, there  exist many edge-disjoint paths from left to right \emph{and} top to bottom with probability converging to one. Strictly speaking, this is the intersection of two separate events, which in general cannot be treated independently. However, the two are increasing events, which we define in Definition~\ref{def:increasing_events}:
For two configurations $\omega$ and $\omega'$, $\omega \leq \omega'$ means that for any $e \in E$ $\omega_e \leq \omega'_e$. 
\begin{definition}[Increasing events]\label{def:increasing_events}
An event $\mathcal{A}$ is called \emph{increasing} if $\omega \leq \omega'$ for some $\omega \in \mathcal{A}$ implies that $\omega' \in \mathcal{A}$.
\end{definition}
Intuitively, an event is increasing if opening more edges does not hurt the probability of a specific configuration belonging to this event. 
For two increasing events $\mathcal{A}$ and $\mathcal{B}$, Harris' inequality applies~\cite{harrisLowerBoundCritical1960}: 
\begin{theorem}
Let $\mathcal{A}$ and $\mathcal{B}$ be two increasing events. Then:
\begin{equation}\label{eq:Harris_inequality}
P(\mathcal{A}\cap \mathcal{B})\geq P(\mathcal{A})\cdot P(\mathcal{B}).
\end{equation}
\end{theorem}
We now cite Lemma 11.22 in Ref.~\cite{grimmettBondPercolationTwo1999}, which precisely deals with the existence of many edge-disjoint paths. Let $M_{\Lambda}$ be the maximum number of edge-disjoint open paths from left to right in a rectangular lattice that is $\Lambda$ edges wide and $(\Lambda-1)$ edges high.
\begin{lemma}\label{lem:edge_disjoint_paths_below_threshold}
Suppose that $p<\frac{1}{2}$. Then there exist strictly positive constants $\beta=\beta(p)$ and $\gamma=\gamma(p)$ such that:
\begin{equation}\label{eq:Mlambda_prob}
    P_p(M_{\Lambda}\leq \beta(\Lambda-1))\leq e^{-\gamma (\Lambda-1)}.
\end{equation}
\end{lemma}
Here, $(M_{\Lambda}\leq \beta (\Lambda-1))$ is a short form for the event $\{\omega \in \Omega \,|\, M_{\Lambda}(\omega) \leq \beta (\Lambda-1)\}$, i.e. all loss configurations $\omega$ for which there exist at most $\beta (\Lambda-1)$ edge-disjoint open paths from left to right in the $\Lambda \times (\Lambda-1)$ rectangular lattice. Consider now a lattice that is $\Lambda-1$ edges wide and $\Lambda$ edges high. Then, Eqn.~\eqref{eq:Mlambda_prob} holds for the event that there are at most $\beta (\Lambda-1)$ open paths from \emph{top to bottom} due to the equivalence of the problem statements under a $90^{\circ}$ rotation of the lattice. 

\subsection{Disjointness and logical support intersection}\label{sec:disjointness}

Disjointness is a property of stabiliser codes first introduced in Ref.~\cite{jochym-oconnorDisjointnessStabilizerCodes2018} to analyse fault-tolerant operations on stabiliser codes. Roughly speaking, a code exhibits high disjointness if there are many disjoint logical representatives. In this section, we will briefly introduce the definition of disjointness and bound the minimum intersection between the support of two anti-commuting logical representatives with its help. We say a qubit $q$ participates in an operator $a \in \mathcal{P}_n$ if $q \in \supp(a)$. We will first define a $c$-disjoint set. Recall that we write $[P]$ for the set of all representatives for an element $P$ in $\Ncurvy(\Scurvy)/\Scurvy$. After loss, however, not all elements in $[P]$ still constitute valid representatives.

We will define the $c$-disjointness $\Delta^c(R)$ for a \emph{subset} $R \subseteq [P]$ of representatives for some $P \in \Ncurvy(\Scurvy)/\Scurvy$ for a code $\Scurvy$\footnote{As opposed to the standard definition from Ref.~\cite{jochym-oconnorDisjointnessStabilizerCodes2018}, where $c$-disjointness is defined for the set of all representatives for any~$P \, \in \, \Ncurvy(\Scurvy)/\Scurvy$.}.  

\begin{definition}[$c$-Disjointness (adapted from Ref.~\cite{jochym-oconnorDisjointnessStabilizerCodes2018}) of a set $R$]\label{def:PRdisjointness}
Let $\Scurvy$ be a stabiliser code.
    A set $P_D \subseteq [P]$ is $c$-disjoint if any qubit only participates in at most $c$ operators $A \in P_D$. We define the $c$-disjointness of an $R \subseteq [P]$, $P \in \Ncurvy(\Scurvy)/\Scurvy$, as:
\begin{equation}
    \Delta^c(R)= \frac{1}{c}\max\{|P_D| \,:\, P_D \subseteq R,\ P_D \text{ is } c\text{-disjoint}\}.
\end{equation}
We define the disjointness of $R$ as $\Delta(R):=\max_{c} \Delta^c(R)$. 
\end{definition}
We will now derive an upper bound on the minimum intersection size of all possible pairs $(X,Z)$, $X \in X_R$, $X_R \subseteq [X]$ and $Z \in Z_R$, $Z_R \subseteq [Z]$. Our Lemma~\ref{lem:intersection_disjointness} is similar to the scrubbing lemma in Ref.~\cite{jochym-oconnorDisjointnessStabilizerCodes2018}, but it does not require knowledge of the distance of $\Scurvy$. 
\begin{restatedlemma}[Upper bound on minimum intersection size]\label{lem:intersection_disjointness}
Let $\{\Scurvy_n\}$ be a family of codes encoding one logical qubit. Let $X_R$ and $Z_R$ be non-empty subsets of $[X]$ and $[Z]$ respectively, defined for each $\Scurvy_n$. Let $\Delta_n(X_R)$ and $\Delta_n(Z_R)$ be the disjointness of $X_R$ and $Z_R$, as defined in Definition~\ref{def:PRdisjointness}, for every code size $n$. Then 
\begin{equation}\label{eq:disjointness_bound_intersection}
    \min_{X \in X_R, Z\in Z_R} \big|\mathrm{supp}(Z) \cap \mathrm{supp}(X)\big| \leq n/(\Delta_n(X_R)\Delta_n(Z_R)).
\end{equation}
\end{restatedlemma}
\begin{proof}
From the definition of disjointness, we know that for any $n$ there exist constants $c_X$ and $c_Z$ and sets of representatives $X_D \subseteq X_R$, $Z_D \subseteq Z_R$ such that $|X_D|=c_X\Delta(X_R)$ and $|Z_D| =c_Z \Delta(Z_R)$, where we dropped the subscript $n$ from $\Delta_n$. Moreover, we know that any qubit $q \in Q$ participates in at most $c_X$ ($c_Z$) elements in $X_D$ ($Z_D$). 
Consider the following sum over the size of the support intersections for all pairs $(X,Z) \in X_D \times Z_D$:
\begin{equation}
H:=\sum_{(X,Z) \in X_D \times Z_D} |\supp(X)\cap \supp(Z)|.
\end{equation}
We can rewrite this sum as a sum over the set of all qubits $Q$ instead:
\begin{equation}
H=\sum_{q \in Q} \sum_{(X,Z) \in X_D \times Z_D} \delta\left(q \in \supp(X) \cap \supp(Z)\right),
\end{equation}
 where $\delta\left(q \in \supp(X) \cap \supp(Z)\right)$ is zero if $q \notin \supp(X)\cap \supp(Z)$ and one otherwise. We can write the latter as a product:
 \begin{equation}
     \delta(q \in \supp(X) \cap \supp(Z))=\delta(q \in \supp(X))\cdot \delta(q \in \supp(Z))
 \end{equation} 
 and further simplify by splitting up the sum into a product of sums over $X_D$ and $Z_D$:
\begin{align}
H&=\sum_{q \in Q} \sum_{(X,Z) \in X_D \times Z_D} \delta\left(q \in \supp(X) \cap \supp(Z)\right)\\
&=\sum_{(X,Z) \in X_D\times Z_D} \sum_{q \in Q} \delta(q \in \supp(X))\cdot \delta(q \in \supp(Z))\\
&=\sum_{q \in Q} \left[\sum_{X \in X_D}\left(\delta(q \in \supp(X))\right)\sum_{Z \in Z_D}\left(\delta(q \in \supp(Z))\right)\right].
\end{align}
 We now utilise the aforementioned fact that each qubit participates in at most $c_X$ ($c_Z$) elements of $X_D$ ($Z_D$), which follows from the definition of disjointness, to obtain:
\begin{align}
 H\leq \sum_{q \in Q}  c_X  c_Z  =n c_X c_Z. \label{eq:H_upperbound}
 \end{align}
 $H$ is a sum over all intersection sizes for each possible pair $(X,Z) \in X_D \times Z_D$. (Note that the intersection must be at least one, otherwise $X$ and $Z$ cannot anti-commute.) Thus, the arithmetic mean $\overline{H}$ of all intersection sizes per pair $(X,Z)$ is $H$ divided by the total number of pairs $(X,Z)$, i.e. $|X_D \times Z_D|$. This quantity is directly related to disjointness via:
 \begin{equation}
     \big|X_D \times Z_D\big|=\big |X_D \big| \cdot \big|Z_D\big|=c_X \Delta(X_R) c_Z\Delta(Z_R),
 \end{equation}
where we inserted the definition of $c$-disjointness (Definition~\ref{def:PRdisjointness}) for the sets $X_D$ and $Z_D$. We conclude that $\overline{H}$ is given by:
 \begin{equation}
 \overline{H}:=\frac{H}{\big | X_D \times Z_D \big|}=\frac{H}{c_X \Delta(X_R) c_Z \Delta(Z_R)}.
 \end{equation}
For any set of non-negative numbers, the arithmetic mean of this set of numbers is greater than or equal to its minimum and thus:
 \begin{align}
     \min_{(X,Z)\in X_D \times Z_D} |\supp(X)\cap \supp(Z)| \leq \overline{H} \leq \frac{nc_Xc_Z}{c_X \Delta(X_R) c_Z \Delta(Z_R)}=\frac{n}{\Delta(X_R) \Delta(Z_R)},
 \end{align}
 where we inserted Eqn.~\eqref{eq:H_upperbound}. Finally, since $X_D \subseteq X_R$ and $Z_D \subseteq Z_R$, the minimum over $X_R \times Z_R$ is at most the minimum over $X_D \times Z_D$, which yields Eqn.~\eqref{eq:disjointness_bound_intersection}.
 
\end{proof}
We will use Lemma~\ref{lem:intersection_disjointness} as the last step in our proof of Theorem~\ref{thm:threshold_surface_informal}.

\subsection{Proof of Theorem~\ref{thm:threshold_surface_informal}}\label{sec:proof_spfthreshold_surface}
 
In this section we prove Theorem~\ref{thm:threshold_surface_informal}. We do so by showing that for any $p<1/2$, there exist subsets $X_R\subseteq [X]$ and $Z_R \subseteq [Z]$ of the logical representatives of the planar surface code whose disjointness grows at least linearly in the lattice size $\Lambda$, i.e. as $\Omega(\sqrt{n})$, which follows from Lemma~\ref{lem:edge_disjoint_paths_below_threshold}. We will then utilise Lemma~\ref{lem:intersection_disjointness} to show that there exists a $g$-SPF pair for any $p < 1/2$, provided that $n$ is large enough. 

First, as discussed in Section~\ref{sec:surface_code}, any path from left to right in the $\Lambda \times (\Lambda-1)$ lattice ($\Lambda$ edges wide and $\Lambda-1$ edges high), or from top to bottom in its dual $(\Lambda-1)\times \Lambda$, corresponds to a representative of the logical $X$- or $Z$-operator, respectively. Our goal is to lower bound the probability of there being \emph{at least} $\beta (\Lambda-1)$ edge-disjoint paths from left to right in the $\Lambda \times (\Lambda-1)$ lattice \emph{and} at least $\beta (\Lambda-1)$ edge-disjoint paths from top to bottom in the dual lattice. We denote the number of left-right edge-disjoint open paths in the lattice by $M_{\mathrm{lr}}$ and the number of top-bottom ones in the dual lattice by $M_{\mathrm{tb}}$. We then have:
\begin{equation}
    P_p\big((M_{\mathrm{lr}}>\beta (\Lambda-1)) \cap(M_{\mathrm{tb}}>\beta (\Lambda-1))\big)\geq P_p\big((M_{\mathrm{lr}}>\beta (\Lambda-1))\big) \cdot P_p\big((M_{\mathrm{tb}}>\beta ( \Lambda-1))\big),
\end{equation}
where we used Harris' inequality (Eqn.~\eqref{eq:Harris_inequality}) for the increasing events $(M_{\mathrm{lr}}>\beta (\Lambda-1)) $ and $(M_{\mathrm{tb}}>\beta (\Lambda-1))$. 
We now use Lemma~\ref{lem:edge_disjoint_paths_below_threshold} and Eqn.~\eqref{eq:Mlambda_prob} to obtain:
\begin{align}
     P_p\big((M_{\mathrm{lr}}>\beta (\Lambda-1)) \cap(M_{\mathrm{tb}}>\beta (\Lambda-1))\big) \geq (1-e^{-\gamma (\Lambda-1)})^2 \geq 1-2 e^{-\gamma (\Lambda-1)}.\label{eq:lr_tb_paths_prob}
\end{align}
Note that we can use the same $\beta$ and $\gamma$ for both $M_{\mathrm{lr}}$ and $M_{\mathrm{tb}}$ due to the symmetry of the problem statement under rotation.

We will now state a corollary that uses Eqn.~\eqref{eq:lr_tb_paths_prob} to bound the probability of there being $\Theta(\sqrt{n})$ edge-disjoint $X$ and $Z$ representatives in the planar surface code \emph{after loss}. After qubit loss, a path from left to right in the $\Lambda \times (\Lambda-1)$ lattice or top to bottom in the $(\Lambda -1) \times \Lambda$ dual lattice without support on the lost qubit set $L$ is a valid logical representative for $X$ or $Z$ respectively. Hence, to find a large number of disjoint logical $X$- or $Z$-representatives \emph{after} loss, it suffices to find a large number of edge-disjoint paths from left to right (top to bottom) in the (dual) lattice that have no support on the lost qubits $L$.  

For a given surface code $\Scurvy_n$, let $X_R$ be the set of $X$-representatives defined by the paths from left to right and $Z_R$ the set of $Z$-representatives defined by the paths from top to bottom in the dual lattice. After loss, only the representatives with no support on the lost qubits remain valid.

For simplification of notation we introduce the function $\mathcal{L}_L:\mathbf{P}(\mathcal{P}_n) \rightarrow \mathbf{P}(\mathcal{P}_n)$ that maps any set of Pauli operators $A$ to its subset of operators with no support on the lost qubits $L$:
\begin{equation}\label{eq:Lldef}
    \mathcal{L}_L(A)=\{a \in A \mid \supp(a) \cap L = \emptyset\}.
\end{equation}
We obtain the following corollary:
\begin{corollary}\label{cor:disjointness_large}
Consider a family $\{\Scurvy_n\}$ of $n$-qubit planar surface codes on a $\Lambda \times (\Lambda-1)$ lattice where qubits are assigned to edges. Let each qubit be lost independently with probability $p$, i.e. the loss configuration $L$ is distributed according to Eqn.~\eqref{eq:loss_prob}. Let $X_R \subseteq [X]$ and $Z_R \subseteq [Z]$ be the sets of path representatives defined above. Then, there exist strictly positive constants $\beta=\beta(p)$ and $\gamma=\gamma(p)$ such that:
\begin{equation}\label{eq:disjointness_small_exp_decay}
    P_p\left((\Delta_n(\mathcal{L}_L(X_R)) > \beta (\Lambda-1)) \cap (\Delta_n(\mathcal{L}_L(Z_R)) > \beta (\Lambda-1))\right)\geq 1-2 e^{-\gamma (\Lambda-1)}.
\end{equation}
\end{corollary}
\begin{proof}
$X_R$ and $Z_R$ correspond to the sets of all paths from left to right in the lattice and from top to bottom in the dual lattice, respectively. After loss, these sets are reduced to $\mathcal{L}_L(X_R)$ and $\mathcal{L}_L(Z_R)$. They correspond to the \emph{open paths} from left to right or top to bottom in the $\Lambda \times (\Lambda-1)$ or $(\Lambda-1) \times \Lambda$ lattice where the edges corresponding to $L$ are closed (see Section~\ref{sec:relation_percolation}). A set of $k$ edge-disjoint open paths yields $k$ pairwise disjoint representatives, i.e. a $1$-disjoint subset of size $k$. Hence, the number of edge-disjoint open paths from left to right or top to bottom lower bounds the disjointness of $\mathcal{L}_L(X_R)$ and $\mathcal{L}_L(Z_R)$, respectively. We can thus identify the event  $(\Delta_n(\mathcal{L}_L(X_R)) > \beta (\Lambda-1))$ with the event  $(M_{\mathrm{lr}}>\beta (\Lambda-1))$ and the event $(\Delta_n(\mathcal{L}_L(Z_R)) > \beta (\Lambda-1))$ with the event $(M_{\mathrm{tb}}>\beta (\Lambda-1))$ and use Eqn.~\eqref{eq:lr_tb_paths_prob} to lower bound the probability of their intersection, arriving at Eqn.~\eqref{eq:disjointness_small_exp_decay}.
\end{proof}
Corollary~\ref{cor:disjointness_large} states that with a probability converging to one the disjointness of the surviving elements of $X_R$ and $Z_R$ scales with $(\Lambda-1)$. 
 The total number of qubits is $n=\Lambda^2+(\Lambda-1)^2 < 2\Lambda^2$, meaning the disjointness approximately scales with $\sqrt{n}$. It only remains to apply Lemma~\ref{lem:intersection_disjointness} to arrive at Theorem~\ref{thm:threshold_surface_informal}. 
 
 Inserting $\Delta_n(\mathcal{L}_L(X_R)) > \beta (\Lambda-1)$,  $\Delta_n(\mathcal{L}_L(Z_R)) > \beta (\Lambda-1)$ and  $n=\Lambda^2+(\Lambda-1)^2  < 2 \Lambda^2$ into Eqn.~\eqref{eq:disjointness_bound_intersection}, we obtain (for $\Lambda >1$):
\begin{equation}\label{eq:upper_bound_intersection_final}
\min_{X \in \mathcal{L}_L(X_R), Z \in \mathcal{L}_L(Z_R)} \big|\supp(Z) \cap \supp(X)\big| < \frac{2 \Lambda^2}{\beta^2 (\Lambda-1)^2}=\frac{2}{\beta^2}\left(\frac{1}{1-(2/\Lambda)+(1/\Lambda^2)}\right).
\end{equation}
As the second factor in the product is monotonically decreasing for $\Lambda \geq 2$, we bound it by 4 for $\Lambda \geq 2$ and finally arrive at:
\begin{equation}\label{eq:upper_bound_intersection_final_2}
\min_{X \in \mathcal{L}_L(X_R), Z \in \mathcal{L}_L(Z_R)} \big|\supp(Z) \cap \supp(X)\big| < 8/\beta^2.
\end{equation}
Set $g:=\lfloor 8/\beta^2 \rfloor$. As the minimum intersection size in Eqn.~\eqref{eq:upper_bound_intersection_final_2} is an integer strictly smaller than $8/\beta^2$, it is at most $g$. Since $\alpha(X,Z) \leq |\supp(X) \cap \supp(Z)|$ for any pair $(X,Z)$, and every $X \in \mathcal{L}_L(X_R)$, $Z \in \mathcal{L}_L(Z_R)$ fulfils the Lost-Qubits condition by construction, the minimising pair is a $g$-SPF pair. Hence:
\begin{equation}\label{eq:eventinclusion}
    (\Delta_n(\mathcal{L}_L(X_R)) > \beta(\Lambda-1)) \cap (\Delta_n(\mathcal{L}_L(Z_R)) > \beta(\Lambda-1)) \subseteq L_{\mathrm{success}}(\leq g),
\end{equation}
where the event $L_{\mathrm{success}}(\leq g)$ is defined in Eqn.~\eqref{eq:loss_config_SPFpairexists}. It is the event that there exists a $g$-SPF pair. By monotonicity of probability, if  $A \subseteq B$ for the events $A$ and $B$, then $P(A) \leq P(B)$ and therefore Corollary~\ref{cor:disjointness_large} gives:
\begin{align}
    P_p(L_{\mathrm{success}}(\leq g)) &\geq  P_p\left((\Delta_n(\mathcal{L}_L(X_R)) > \beta(\Lambda-1)) \cap (\Delta_n(\mathcal{L}_L(Z_R)) > \beta(\Lambda-1))\right) \label{eq:final_line_1}\\
    &\geq  1-2 e^{-\gamma (\Lambda-1)},
\end{align}
which converges to one for $n \rightarrow \infty$ (equivalently $\Lambda \rightarrow \infty$), thereby proving the first statement in Theorem~\ref{thm:threshold_surface_informal}.  

To see that for $p>1/2$ the $g$-SPF success rate converges to zero, we observe that localisation via $g$-SPF requires the survival of logical information. The erasure threshold from Ref.~\cite{staceThresholdsTopologicalCodes2009} is exact. That is, for $p>1/2$ logical information is destroyed with a probability converging to one. Hence, the last statement in Theorem~\ref{thm:threshold_surface_informal} is proved. This completes the proof of Theorem~\ref{thm:threshold_surface_informal}.
It should be noted that, technically, disjointness is not needed for the proof. Instead, we could have simply spoken of sets of edge-disjoint operators. However, we chose to purposefully illuminate the connection between low intersection and disjointness.

\section{Optimisation algorithms}\label{sec:opt_algos}

In this section we introduce our optimisation algorithms to solve Target- and $g$-SPF. We formulate both problems as integer linear programs (ILPs) and use an ILP solver to solve it. This constitutes our algorithm D-LoFi for Target- and $g$-SPF. Moreover, we cast $g$-SPF as a decoding problem and use a decoder to construct a fast heuristic solver, dubbed H-LoFi. Both algorithms use the symplectic representation for stabiliser codes and share the same pre-processing step to handle the Lost-Qubits condition. We discuss the pre-processing and its formalisation into linear algebra tasks in Section~\ref{sec:prepro}. In Section~\ref{sec:algo_deterministic}, we discuss D-LoFi for both target and $g$-SPF. Finally, in Section~\ref{sec:algo_heuristic} we will explain H-LoFi for $g$-SPF.  

In what follows, we often write $X$ and $Z$ to denote symplectic vectors in $\Ftwotwon$ representing logical $X$ and $Z$ operators; thus we do not specify the difference between the space $\Ftwotwon$ and $\mathcal{P}_n$ and use them interchangeably. Moreover, we write $v \in \mathbb{F}_2^{l\times 1}$ to denote an $l \times 1$ column vector in the vector space $\mathbb{F}_2^l$; and $w \in \mathbb{F}_2^{1 \times l}$ for a vector in the dual space of $\mathbb{F}_2^l$, which can be interpreted as a $1 \times l$ row vector.  

\subsection{Pre-processing for $g$-/Target-SPF}\label{sec:prepro}

Given a set of lost qubits $L$, the pre-processing step allows the update of the stabiliser tableau $T$ to represent the \emph{reduced} stabiliser code describing the state after loss. Similarly, we choose some representatives for $X$ and $Z$ before qubit loss and find new representatives $X'$ and $Z'$ that have no support on the lost qubits. This creates the tableau of the punctured code where all qubits in $L$ have been removed. This tableau is then passed onto the ILP solver or the decoder. All pre-processing amounts to linear algebra over $\mathbb{F}_2$. In this section, we explain how the Lost-Qubits condition (Definition~\ref{def:Lemptyset}) translates to finding a solution to a system of linear equations. 

 Given a stabiliser group $\Scurvy$ experiencing the loss of the qubits in the set $L$, we find the updated stabiliser group $\Scurvy'$ that has no support on the lost qubits:
\begin{definition}[Reduced stabiliser group]
A stabiliser group $\Scurvy$ is updated to $\Scurvy'$ after losing qubits in $L$ by removing all $s \in \Scurvy$ with $\supp(s)\cap L \neq \emptyset$. That is, $\Scurvy'=\mathcal{L}_L(\Scurvy)$, where $\mathcal{L}_L$ is defined by Eqn.~\eqref{eq:Lldef}.
\end{definition}
The goal is to find a new stabiliser tableau describing the updated $\Scurvy'$. However, before we find that, we check for logical loss. Logical loss is defined as follows:
\begin{definition}[Logical loss]\label{def:logical_loss}
    Let $\Scurvy$ be a stabiliser code. Let $L$ be the set of lost qubits. Logical information of at least one qubit is lost if $L$ contains the full support of a logical representative for some non-trivial $P \in \Ncurvy(\Scurvy)/\Scurvy$.
\end{definition}
Whether or not logical information has been lost can be easily checked by attempting to clean a representative $X \in [X]$ or $Z \in [Z]$.
\begin{lemma}[Cleaning a logical]
Let $\Scurvy$ be a stabiliser code defining one logical qubit.
Logical information has been lost if and only if there exists a non-trivial $P \in  \Ncurvy(\Scurvy)/\Scurvy$ so that no $s \in \Scurvy$ fulfils $\supp(\tilde{P}s)\cap L=\emptyset$ for every single $\tilde{P} \in [P]$.
\end{lemma}
\begin{proof}
We first show the forward direction. Assume $L$ contains the support of some representative of a non-trivial $P \in  \Ncurvy(\Scurvy)/\Scurvy$. Then, each representative of $P$ either anti-commutes with every $Z \in [Z]$ or $X \in [X]$ or both. Without loss of generality, let each representative of $P$ anti-commute with every $Z \in [Z]$. This implies that $Z$ and the representative $\tilde{P}$ must intersect and thus $\supp(Z) \cap L \neq \emptyset$ $\forall Z \in [Z]$. We show the backward direction with the help of the cleaning lemma in Ref.~\cite{bravyiNogoTheoremTwodimensional2009a}, which states that a subset $L$ of qubits either contains the full support of a non-trivial logical or every logical representative $\tilde{P} \in [P]$ can be multiplied with an $s \in \Scurvy$ such that $\supp(\tilde{P}s) \cap L =\emptyset$, i.e. any logical representative can be `cleaned' out of the region. Assume, w.l.o.g., one cannot clean $Z$ out of $L$. Then, per the cleaning lemma, $L$ must contain the support of a non-trivial logical.
\end{proof}
In the pre-processing, we pick some representatives for $X$ and $Z$ and attempt to clean both of them from $L$. When this is not possible, logical information is lost.  We now discuss how to find a clean $X$ or $Z$ using the symplectic representation. 

Given an $r  \times (2n)$ stabiliser tableau $T$, a set of lost qubits $L \subseteq Q$ and a representative $X \in [X]$, we want to find another representative $X'$ that has no support on the lost qubits. 
To do so we need to find a stabiliser $s$ we can multiply $X$ with such that $Xs$ has no support on $L$, meaning we must find a solution to:
\begin{equation}\label{eq:cleanlogical}
(cT+X)[L]=0,
\end{equation}
where $c$ and $X$ are $1 \times r$ and $1 \times (2n)$ row vectors respectively. After finding a solution $c$ to Eqn.~\eqref{eq:cleanlogical}, we can set $X':=cT+X$ and use the reduced punctured logical, i.e. $X' \mapsto X'[Q-L]$.
One can find a cleaned $Z$-logical analogously. An advantage of the pre-processing is that Eqn.~\eqref{eq:cleanlogical} immediately flags when no clean logical can be found and thus all logical information has been lost. When it has no solution, our algorithms stop early and return a failure flag.  

It remains to update the stabiliser group $\Scurvy$ by removing all $s \in \Scurvy$ with support on $L$. To do so we map $T$ to a $T'$ where we keep only the linear combinations of rows of $T$ with no support on $L$.
For this, we find the solution space of the following linear equation:
\begin{equation}\label{eq:clean_T}
    (cT)[L]=0.
\end{equation}
For any solution $c$ the vector $cT$ will have no support on $L$. Therefore, to update $T$ after loss, we first find the solution space $\mathcal{C}_{L}$ of Eqn.~\eqref{eq:clean_T}, define the matrix $C_{L}$ whose rows are a basis of $\mathcal{C}_{L}$ and update $T$ after loss:
\begin{equation}
    \mathcal{C}_{L}=\{c \in \mathbb{F}_2^{1\times r}: (cT)[L]=0\}, \quad  T \mapsto T':= (C_LT)[Q-L].
\end{equation}
Note that $T'$ now has dimension $r'  \times 2(n-|L|)$, where $r'$ is the dimension of $ \mathcal{C}_{L}$.
After cleaning the logicals $X$ and $Z$ and updating the stabiliser tableau, the \emph{reduced} tableau $T'$ and the reduced representatives $X'$ and $Z'$ are handed to the deterministic or heuristic algorithm. We summarise the pre-processing in Algorithm~\ref{algo:prepro} in Appendix~\ref{app:algos}.

\subsection{Deterministic algorithm for Target-SPF and $g$-SPF}\label{sec:algo_deterministic}

In this section we describe a deterministic algorithm for solving the Target- and $g$-SPF problems. Following the pre-processing step that encodes the Lost-Qubits condition, we encode the Target condition and the $g$-condition, as defined in Section~\ref{sec:formalisation}, as constraints to find valid measurement patterns that localise the logical information. We then formulate Target-SPF and $g$-SPF as optimisation problems using these constraints to find logical operators with the smallest support. We solve these optimisation problems by mapping them to integer linear programs and using an ILP solver. 

Given a loss configuration $L$, an $r \times 2n$ stabiliser tableau $T$ and two representatives $X \in \mathbb{F}_2^{1 \times (2n)}$ and $Z \in \mathbb{F}_2^{1 \times (2n)}$ of the logical operators, we first apply the pre-processing from Algorithm~\ref{algo:prepro} to handle the Lost-Qubits condition from Definition~\ref{def:Lemptyset}. If no logical loss occurs, we obtain the cleaned and punctured $r' \times 2(n-|L|)$ tableau $T'$ and the cleaned and punctured representatives $X_0,Z_0$ $\in \mathbb{F}_2^{1 \times (2(n-|L|))}$. We can thus represent any other clean and punctured representative $X$ and $Z$ in symplectic notation as row vectors in $\mathbb{F}_2^{1\times (2(n-|L|))}$ via: 
\begin{align}
X &=X_0 + B_xT',\\
Z &= Z_0 + B_zT',
\end{align}
where $B_x$, $B_z \in \mathbb{F}_2^{1\times r'}$ and $X_0$, $Z_0 \in \mathbb{F}_2^{1\times (2(n-|L|))}$. Note that we now work with a new qubit index set with indices from 0 to $(n-|L|)-1$. The target qubit $t$ is re-indexed.
The Target condition from Definition~\ref{def:tanti} requires that $X$ and $Z$ valid for SPF only anti-commute qubit-wise on the target qubit $t$. Therefore, 
\begin{align}\label{eq:target_constraint}
 \langle X[t], Z[t] \rangle =1, \quad \sum_{i \neq t} \langle X[i], Z[i] \rangle =0.
\end{align}
Recall that we write $\langle \cdot ,\cdot \rangle$ for the symplectic product. All representatives that satisfy the Target condition (Definition~\ref{def:tanti}) or equivalently constraint~\eqref{eq:target_constraint} are valid for localisation of the logical information. However, out of all valid pairs, we find a pair of representatives that minimises the total support to mitigate future losses and possible measurement errors. Therefore, Target-SPF (see Problem~\ref{prob:tSPF}) can be summarised as the following optimisation problem:
 \begin{align}\label{eq:tspf_optimisation}
\min_{B_x, B_z \in \mathbb{F}_2^{1\times r'}} \hspace{2em} & \big|\supp(X_0 + B_xT') \cup  \supp(Z_0 + B_zT')\big|\\
\text{s.t.} \hspace{2em} 
    &   \langle X[t], Z[t] \rangle =1 \notag,\\
    & \sum_{i \neq t }^{n-|L|-1} \langle X[i], Z[i] \rangle = 0 \notag.
\end{align}
The sums are over $\mathbb{Z}$, whereas $\langle \cdot,\cdot\rangle$ uses summation over $\mathbb{F}_2$.
Similarly, the $g$-condition (Definition~\ref{def:gcondition}) for $g$-SPF (Problem~\ref{prob:gSPF}) can be written as:
\begin{equation}\label{eq:fewerganticommute}
\sum_{i=0}^{n-|L|-1} \langle X[i] , Z[i] \rangle \leq g,
\end{equation}
where the sum goes over $\mathbb{Z}$. Eqn.~\eqref{eq:fewerganticommute} ensures that $X$ and $Z$ anti-commute qubit-wise on at most $g$ qubits. Moreover, we can also minimise the support of representative logical operators to construct the following optimisation problem for $g$-SPF:
\begin{align}\label{eq:gspf_optimisation}
\min_{B_x, B_z \in \mathbb{F}_2^{1\times r'}} \hspace{2em} &   \big|\supp(X_0 + B_xT') \cup  \supp(Z_0 + B_zT')\big|\\
\text{s.t.} \hspace{2em} 
    & \sum_{i=0}^{n-|L|-1} \langle X[i], Z[i] \rangle \leq g\notag.
\end{align}
We solve the optimisation problems by mapping them to an integer linear program (ILP). However, the constraints for the qubit-wise (anti-)commutation of logical operators (Eqn.~\eqref{eq:target_constraint} and Eqn.~\eqref{eq:fewerganticommute}) are quadratic constraints over binary variables.  
Therefore, we convert the binary quadratic constraints into integer linear constraints by introducing additional constraints and variables. First, we linearise the quadratic terms by introducing auxiliary binary variables and additional linear constraints. More specifically, we replace each product $uv$ by a binary variable $w$ and add the constraints $w\leq u$, $w\leq v$, $w\geq u+v-1$. Second, we convert each sum modulo 2 into an equality over integers by introducing a non-negative integer variable and a binary variable. For example, for each sum $s=a_1+\dots+a_m \bmod 2$, we introduce a non-negative integer variable $y$ and impose $a_1+\dots+a_m = s+ 2y$. 

The resulting problem can then be solved using an ILP solver, for example Google's CP-SAT~\cite{ortools,cpsatlp}.
We summarise our deterministic solver D-LoFi in Algorithm~\ref{algo:det} in Appendix~\ref{app:algos}.

\subsection{Heuristic algorithm: H-LoFi}\label{sec:algo_heuristic}

In this section we explain how we formulate $g$-SPF, as defined in Problem~\ref{prob:gSPF}, as a decoding problem, which can then be solved heuristically using any decoder. We will then explain our heuristic algorithm H-LoFi that solves $g$-SPF.

 H-LoFi is divided into two parts. In the first part it finds a representative $Z$ of low (but not per se minimal) support that fulfils the Lost-Qubits condition and in the second part it finds a representative $X$ of low support that anti-commutes with the previously found $Z$ on a small number of qubits and also fulfils the Lost-Qubits condition. We use low support of both $X$ and $Z$ as a heuristic for finding a good SPF pair $(X,Z)$ for two reasons. Firstly, the shorter the supports of $X$ and $Z$, the fewer qubits on which they can overlap, and hence the fewer qubits on which they can anti-commute. Secondly, our deterministic algorithm finds the SPF pair with minimum support. Hence, in order to obtain similar solutions with H-LoFi, we use short support of $X$ and $Z$ individually.  

 The Lost-Qubits condition is encoded into the pre-processing. After applying the pre-processing, H-LoFi works with reduced and punctured stabiliser tableaus and logical representatives. The latter do not contain the indices of the lost qubits and hence the following decoding formulation need not consider any lost qubits. 
 
 To find a low weight $Z$, H-LoFi solves a decoding problem, where an arbitrary initial $X$ representative is appended to the stabiliser group as a new stabiliser, and the decoder then searches for an error that commutes with all stabiliser elements except this particular $X$ representative. In the second part, when it tries to find a low weight $X$ representative that anti-commutes with $Z$ on a small number of qubits, we add the previously found $Z$ to the original stabiliser group, and let the decoder search for an error that commutes with all stabilisers except $Z$. We formulate the decoding problem so that solutions $X$ that anti-commute qubit-wise with $Z$ on a large number of qubits become unlikely. Both decoding problems are defined with respect to error weights and the task is to find the most likely error based on the error weights. We therefore use a decoder that can solve such `most-likely-error' problems. 
We note that a heuristic algorithm for Target-SPF can also be implemented  by first finding a short $Z$ (by decoding) and then finding an $X$ that anti-commutes qubit-wise with $Z$ only on qubit $t$ by solving a system of linear equations. However, this approach did not lead to good performance in a first basic implementation and we left out its discussion in this section and Section~\ref{sec:numerics}.

 We first discuss the formulation as a decoding problem in Section~\ref{sec:decoding_problem} and then explain the exact implementation of the problem in Section~\ref{sec:decoding_solution}. In the following, we will use the word `anti-commute' somewhat sloppily and write $t \in \mathbb{F}_2^{2n}$ anti-commutes with some $e \in \mathbb{F}_2^{2n}$ if we mean that $\langle t, e \rangle=1$ and thus the associated Pauli operators anti-commute.

 \subsubsection{Formulation as a decoding problem}\label{sec:decoding_problem}

 Let $M$ be a $\mu \times \nu$ matrix. We introduce the notation $(M,v)^T$ to denote that we vertically append a $1 \times \nu$ row-vector $v$ to $M$. 
 
 H-LoFi takes as an input an $r  \times 2n$ stabiliser tableau $T$ and a set of lost qubits $L$. It first finds a symplectic basis $\{(X_i,Z_i)\}$ for the logical space. As we only consider codes with one logical qubit, the symplectic basis amounts to a single pair $(X,Z)$. It then applies our pre-processing Algorithm~\ref{algo:prepro} to $T$ and the representatives $X$ and $Z$, yielding either logical information loss, in which case it returns failure, or an updated reduced $r'  \times 2(n-|L|)$  stabiliser tableau $T'$ and cleaned logical representatives $(X',Z')$. The reduced tableau and representatives are only defined on the surviving qubits. Hence, the solution SPF pair of H-LoFi will automatically fulfil the Lost-Qubits condition. 
 If we were to consider the reduced and punctured stabiliser tableau $T'$ of size $r'  \times 2(n-|L|)$ as a new code, it can generally contain additional logical qubits beyond the original one. These additional degrees of freedom do not carry any logical information from the original code, but need to be considered to properly define the decoding problem H-LoFi is based on.
H-LoFi therefore finds a symplectic basis for $\ker(T'\Omega)/\mathcal{R}(T')$ such that $(X',Z')$ is an element of it. We elaborate on how to do so in Appendix~\ref{app:misc}. 

For each symplectic basis pair $(X_j,Z_j)\neq(X',Z')$ it appends $X_j$ and $Z_j$ to $T'$, and finally $X'$:
\begin{equation}\label{eq:append_sympl}
   T' \gets \begin{pmatrix} T'\\ X_j\\ Z_j \end{pmatrix}
   \quad\text{for all } (X_j,Z_j)\neq(X',Z'), \qquad
   T' \gets \begin{pmatrix} T'\\ X' \end{pmatrix}.
\end{equation}
The first step is to replace the arbitrary $Z'$, which may be an operator with large support, with a \emph{short} $Z' \in [Z']$.  To do so H-LoFi solves the following decoding problem:
\begin{equation}\label{eq:decoding_short_Z}
   \mathrm{argmax}_{e \in \mathbf{S}} q(e) | \mathbf{S}=\{e \in \mathbb{F}_2^{2(n-|L|)} |  T'\Omega e=(0,0,0,\ldots,1)^T\},
\end{equation}
where $q(e)$ are \emph{error weights}, which normally determine how likely an error $e$ is; however, here we do not consider them to be probabilities. The solution to Eqn.~\eqref{eq:decoding_short_Z} yields an operator that commutes with all stabilisers defined by the rows of $T'$ and all other newly created logical operators given by all $(X_j,Z_j)$ in the symplectic basis except for $(X',Z')$, and anti-commutes with $X'$. These anti-commutation relations uniquely define the representative class $[Z']$. Therefore, the solution of Eqn.~\eqref{eq:decoding_short_Z} yields a short $Z'$, as desired. We simply would like to steer the decoder in Eqn.~\eqref{eq:decoding_short_Z} to a low-weight solution $Z'$ and hence assign each single-qubit Pauli error $P_i \in \{X_i,Z_i,Y_i\}$ an equal weight $q$ per qubit $i$. 

Having found a short $Z'$, H-LoFi proceeds to find a short $X'$ that fulfils the $g$-condition. To do so, the initial reduced $T'$ is modified similarly to Eqn.~\eqref{eq:append_sympl}, except that after appending the symplectic basis, H-LoFi appends the \emph{newly found short $Z'$} instead of $X'$: $T'\gets (T',Z')^T$. 
The weights in the decoding problem are modified to punish a solution that anti-commutes with $Z'$ on a large number of qubits. To do so, we start out with equally distributed weights. Then, for each qubit $i$ in $\supp(Z')$ we map the weight $q^P[i]$  of a single-qubit Pauli error $P \in \{X,Y,Z\}$ on qubit $i$ to:
\begin{align}
    q^P[i] \mapsto \begin{cases} q^-  & \text{if} \quad Z'[i] \neq P   \\
    q^P[i] &\text{if} \quad Z'[i] = P
    \end{cases}, \label{eq:q_punish}
\end{align}
where $q^-$ is some fixed small value, e.g.~$10^{-3}$.
That is, we make all single-qubit errors that anti-commute with the $i$-th Pauli component of $Z'$ unlikely.
In the end we solve the decoding problem defined in Eqn.~\eqref{eq:decoding_short_Z} but with the updated weights $q$ and the updated tableau $T'$, yielding an $X'$ that fulfils the Lost-Qubits condition and hopefully the $g$-condition. We reject any solution $X'$ for which $\alpha(X',Z')>g$. The entire algorithm repeats for a fixed number of iterations as long as no $g$-SPF pair has been found, where the equal weights $q^P:=q$ are drawn at random during every iteration in order to find new solutions. We summarise the entire heuristic algorithm H-LoFi in Algorithm~\ref{algo:gspf_heuristic} in Appendix~\ref{app:algos}.

\subsubsection{Solving the decoding problem}\label{sec:decoding_solution}
Here we specify how we pass the decoding problems to the actual decoder. We use the Belief-Propagation Ordered-Statistics (BP-OSD) decoder from the Python ldpc package~\cite{Roffe_LDPC_Python_tools_2022,roffeDecodingQuantumLowdensity2020}, which takes an error weight per \emph{index} of the \emph{symplectic} vector, meaning one can assign individual uncorrelated weights to the computational basis vectors of $\mathbb{F}_2^{2(n-|L|)}$ but not to errors $e$ with Hamming weight larger than one. As a result, when using symplectic representation, a single-qubit $Y$-error, e.g.~$e=(1,0,\ldots,0|1,0,\ldots,0)^T$, becomes less likely. However, for our purposes, we want an error model where single-qubit $Y$-operators are not seen as unfavourable, as there might be a solution to $g$-SPF consisting of a pair that contains many $Y$-operators on single qubits which a decoder using the symplectic representation might fail to find. Hence, instead of the symplectic representation of Pauli operators, we use the three-block representation~\cite{websterDistanceFindingAlgorithmsQuantum2026}, a generic decoder modification method meant to treat $Y$ errors on equal footing. We will very briefly define the three-block representation and directly state the decoding problem in three-block representation and weights we pass to the BP-OSD decoder; for more information see for example Ref.~\cite{websterDistanceFindingAlgorithmsQuantum2026}. In the three-block representation, each $P \in \mathcal{P}_1/\langle \pm i \mathbf{I} \rangle$ is mapped to an element of $\mathbb{F}_2^{3}$: $(x|z|x+z)$.
 
For example, the $X$ operator is mapped to $(1|0|1)$. Each single-qubit Pauli operator is thus mapped to an element of the even-weight subspace $H$ of $\mathbb{F}_2^{3}$. Accordingly, a multi-qubit $P \in \mathcal{P}_n$ with symplectic representation $(x|z) \in \mathbb{F}_2^{2n}$ is mapped to $(x|z|x+z) \in H^{\oplus n} \subset \mathbb{F}_2^{3n}$. The product of two elements in $\mathcal{P}_n$ still corresponds to addition in $\mathbb{F}_2^{3n}$. 
For integers $r'$ and $n'$, we introduce the map $R_{2 \mapsto 3}:\mathbb{F}_2^{r' \times 2n'} \rightarrow \mathbb{F}_2^{r' \times 3n'}$ that maps an $r' \times (2n')$ symplectic tableau $T$ to its three-block representation.  
The decoding problem in Eqn.~\eqref{eq:decoding_short_Z} that has to be solved to find a short representative $Z'$ becomes:
\begin{equation}\label{eq:decoding_short_Z_threeblock}
    \mathrm{argmax}_{e \in \mathbf{S}} q(e) \, | \, \mathbf{S}=\{e \in \mathbb{F}_2^{3(n-|L|)} | R_{2 \mapsto 3}(T')e=(0,0,0,\ldots,1)^T\}.
\end{equation}
Using the three-block representation, one can now assign the weights $q^P[i]$ to computational basis vectors in $\mathbb{F}_2^{3(n-|L|)}$. As a result, equal weights $q^P[i]$ for every $P$ $\in \{X,Y,Z\}$ and later the modified weights defined by Eqn.~\eqref{eq:q_punish} can be passed to the BP-OSD decoder. 

We did not insert any matrix analogous to $\Omega$ in Eqn.~\eqref{eq:decoding_short_Z_threeblock} that switches the positions of $x$- and $z$-components. That is because a solution $e$ cannot as straightforwardly be mapped to a Pauli operator as in the case of symplectic representation. We let $e$ be a general element of $\mathbb{F}_2^{3(n-|L|)}$. Let $e[i]$ of a solution $e \in \mathbb{F}_2^{3(n-|L|)}$ denote the restriction of $e$ onto the components of qubit $i$, analogously to the notation used for symplectic representation. Let $e[i]=(e_X|e_Z|e_Y)$. Then, per qubit, $e[i]$ is mapped to its symplectic representation $e'[i] \in \mathbb{F}_2^{2}$ via the map $\lambda$:  
\begin{equation}\label{eq:lambda}
    \lambda(e[i])=e_X(0|1)+e_Z(1|0)+e_Y(1|1).
\end{equation}  
The reader can check that this way the components of $e$ can be interpreted as indicators for which Pauli operators the operator $e$ anti-commutes with\footnote{The map $\lambda$ in Eqn.~\eqref{eq:lambda} can be derived by considering the standard matrix product over $\mathbb{F}_2$ of some stabiliser $t \in \mathbb{F}_2^{1 \times (3n)}$ in three-block notation with some $e  \in \mathbb{F}_2^{(3n) \times 1}$. Let us consider one qubit. As one can easily check, an error $e=(1|0|0)^T$ fulfils $t \cdot e=1$ for some stabiliser $t$ in three-block notation if and only if the $x$-component of $t$ is 1, i.e. if $t=(1|0|1)$ or $t=(1|1|0)$. Thus, we know that $e=(1|0|0)^T$ anti-commutes with the single-qubit Paulis $X$ and $Y$, making it a $Z$-error. Similarly, we can argue that $e=(0|1|0)^T$ and $e=(0|0|1)^T$ correspond to an $X$ and a $Y$-error respectively. Therefore, Eqn.~\eqref{eq:lambda} follows.
}.

\section{Numerics and discussion}\label{sec:numerics}

In this section, we present a numerical analysis of our deterministic and heuristic algorithms. We first use our algorithms to show the existence of a $g$-SPF threshold near $p=1/2$ for the rotated surface code,\footnote{We choose the rotated surface code due to its higher encoding efficiency, making it easier to simulate. Note that the arguments made in Section~\ref{sec:threshold_surfacecode} carry over to the rotated surface code: identifying qubits with the edges between the two checks that contain them yields an analogous rectangular lattice, on which logical representatives are again boundary-to-boundary paths and lost qubits are closed edges.} consistent with Theorem~\ref{thm:threshold_surface_informal}. We then compare the runtime of D-LoFi and H-LoFi, and the minimisation performance of H-LoFi against D-LoFi in terms of the size of the total support of the SPF pair solution $(X,Z)$. Furthermore, we analyse the performance of our algorithms as the code distance (system size) is increased. 

We then reproduce the existence of loss thresholds for Target-SPF on two graph codes from Ref.~\cite{morley-shortLosstolerantTeleportationLarge2019}. More specifically, we consider the triangular and crazy graph code, shown in Fig.~\ref{fig:channels}. Additionally, we compare the runtime of our algorithm against the runtime of the algorithm in Ref.~\cite{morley-shortLosstolerantTeleportationLarge2019} for the crazy graph code. All numerical simulations in this section were performed on the DelftBlue supercomputer at TU Delft~\cite{DHPC2024}. We use Google's CP-SAT solver from Ref.~\cite{ortools,cpsatlp} for D-LoFi and the BP-OSD decoder from the Python ldpc package from Ref.~\cite{roffeDecodingQuantumLowdensity2020,Roffe_LDPC_Python_tools_2022}.

We will approximate the success rate of $P_p(L_{\mathrm{success}}(\leq g))$ for $g$-SPF in the presence of loss in the following way: For every $p$ we will sample $N_{L}$ loss configurations according to the probability distribution $P_p(L)$ in Eqn.~\eqref{eq:loss_prob}. We will check for which loss configurations our algorithms find a $g$-SPF pair and divide the number of successful instances by $N_L$, i.e.:
\begin{equation}\label{eq:P_p_estimator}
    P_p(L_{\mathrm{success}}(\leq g)) \approx \frac{\sum_{L \sim P_p(L)} \delta_L(g\text{-SPF pair exists for }L) }{N_L},
\end{equation}
where $\delta_L(g\text{-SPF pair exists for }L)$ is zero if no $g$-SPF pair exists for the loss configuration $L$ and one if it exists. While D-LoFi is exact and thus applying Eqn.~\eqref{eq:P_p_estimator} to its output is an unbiased estimator for the actual $g$-SPF success rate that converges to the real rate for $N_L\to \infty$, we do not have such guarantees for H-LoFi. However, the point is to approximate the real success rate of perfect $g$-SPF accurately with either algorithm.

In the case of Target-SPF, we, following Ref.~\cite{morley-shortLosstolerantTeleportationLarge2019}, assume that the target and input qubit are never lost, i.e. the loss configuration $L$ is only sampled as a subset of $Q - \{i,t\}$, where $i$ and $t$ stand for the input and the target qubit.  $P_p: \mathbf{P}(Q - \{i,t\}) \rightarrow [0,1]$ is then given by a slight modification of Eqn.~\eqref{eq:loss_prob}, namely $ P_p(L)= p^{|L|}(1-p)^{n-|L|-2}$ for any $L \in \mathbf{P}(Q - \{i,t\})$. We then define $L_{\mathrm{success}}(t)$ and the estimator for the success rate $P_p(L_{\mathrm{success}}(t))$ as:
\begin{align}
    P_p(L_{\mathrm{success}}(t)) &\approx \frac{\sum_{L \sim P_p(L)} \delta_L(\text{Target-SPF pair exists for }L) }{N_L}, \label{eq:P_p_estimator_target}\\
    L_{\mathrm{success}}(t)&:=\{L \in \mathbf{P}(Q - \{i,t\})\big | \text{Target-SPF pair exists for }L \, \text{on target qubit } t\}.
\end{align}
 Note that  $L_{\mathrm{success}}(t)$ and $L_{\mathrm{success}}(\leq g)$ use similar notation, but differentiate between the events that a \emph{target}-SPF pair exists and that a $g$-SPF pair exists. In the latter case the pair is only required to anti-commute qubit-wise on at most $g$ qubits, whereas in the former case they must anti-commute on exactly the target qubit.

Even though minimising the support size of the SPF pair $(X,Z)$ is not necessary for computing the threshold, we chose to design D-LoFi such that it minimises the support of the SPF pair and H-LoFi such that it prefers solutions with small total support.

\begin{figure*}[t]
    \begin{minipage}{0.49\textwidth}
        \includegraphics[width=0.95\linewidth]{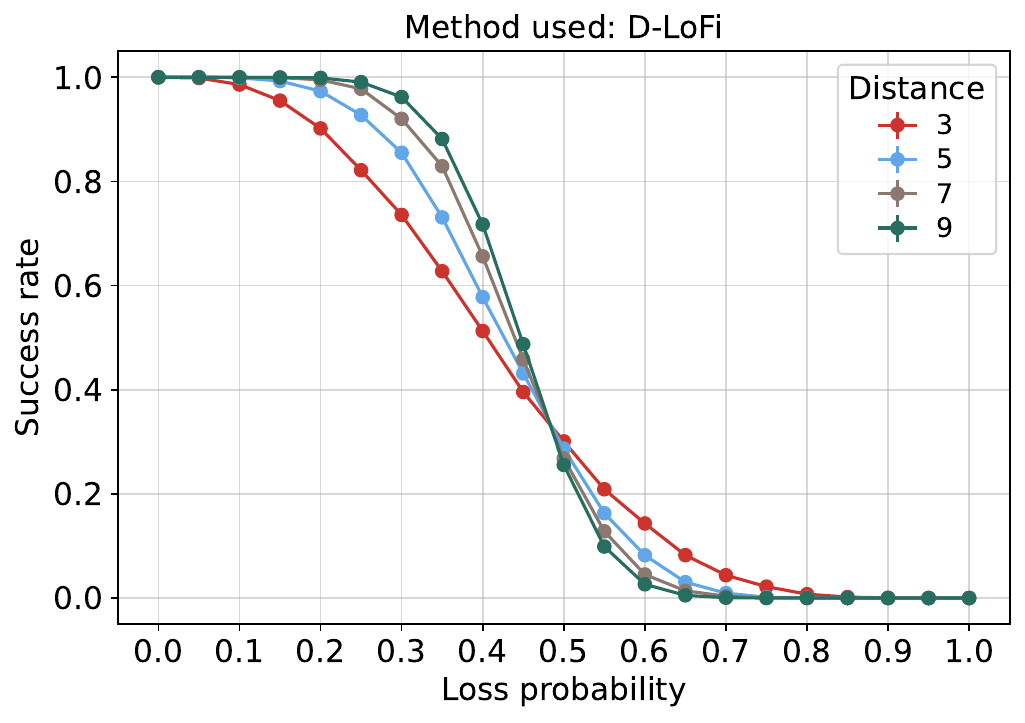}
        \subcaption{Method: D-LoFi}
    \end{minipage}
    \begin{minipage}{0.49\textwidth}
        \includegraphics[width=0.95\linewidth]{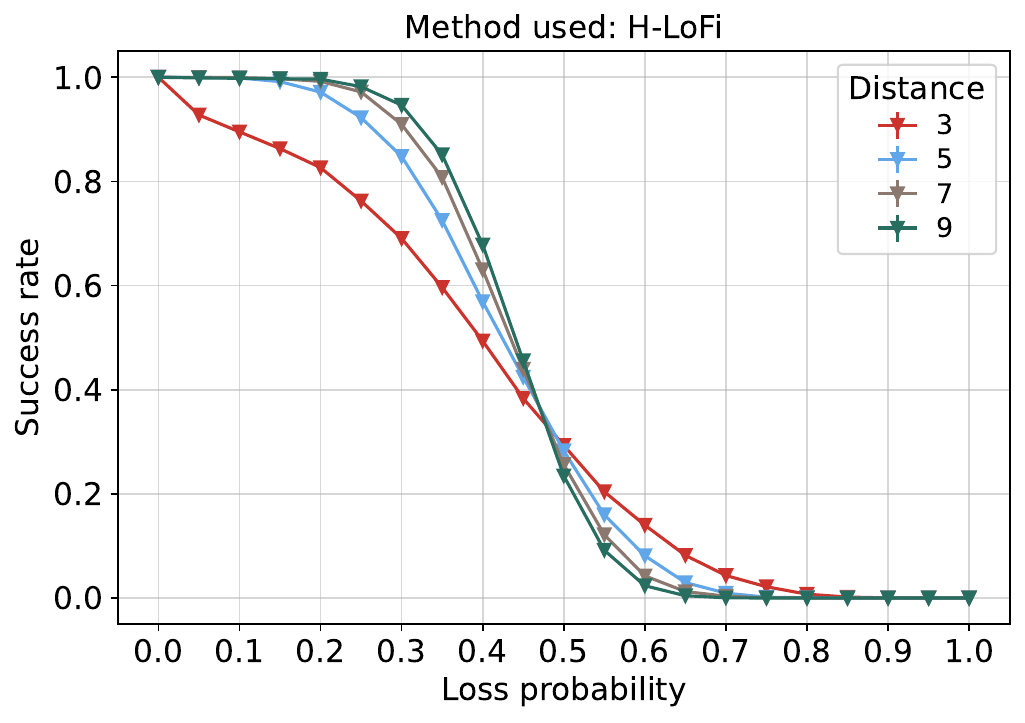}
        \subcaption{Method: H-LoFi}
    \end{minipage}
    \caption{$g$-SPF threshold, where $g=1$, for the rotated surface code. For every loss probability we sample at least 5000 loss configurations according to an i.i.d. loss probability distribution and estimate the success rate according to Eqn.~\eqref{eq:P_p_estimator}. The standard error of the mean is too small to be visible in the figure and has been omitted for that reason.}
    \label{fig: sc_thres}
\end{figure*}

\begin{figure*}[h!]
    \begin{minipage}{0.49\textwidth}
        \includegraphics[width=0.95\linewidth]{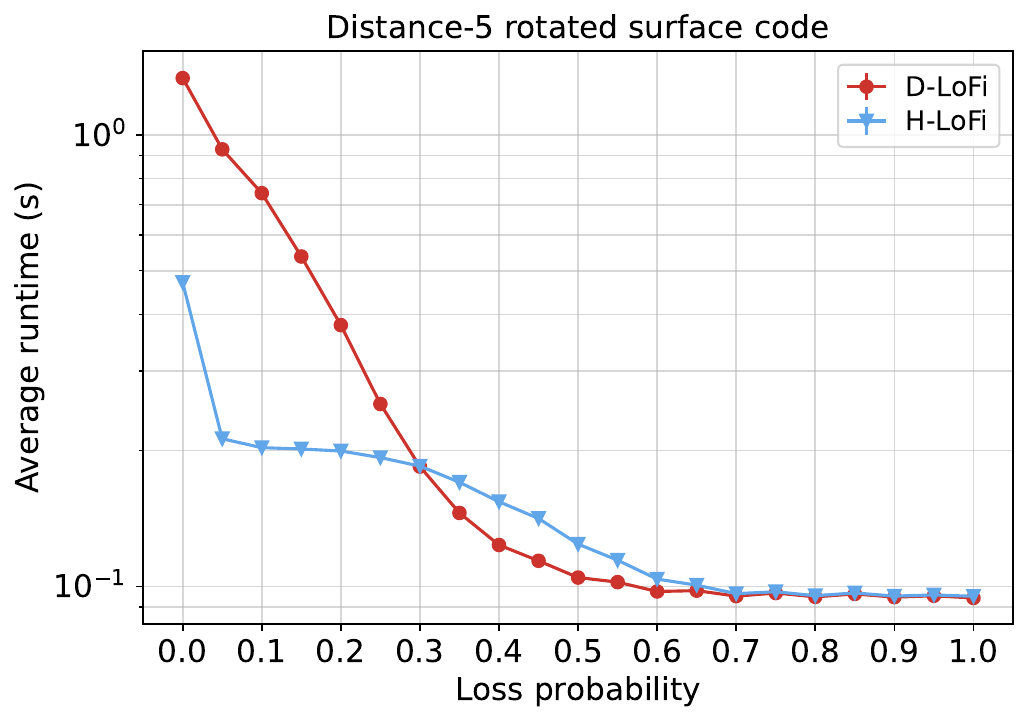}
        \subcaption{}
    \end{minipage}
    \begin{minipage}{0.49\textwidth}
        \includegraphics[width=0.95\linewidth]{p_0.1.pdf}
        \subcaption{}
    \end{minipage}
    \begin{minipage}{0.49\textwidth}
        \includegraphics[width=0.95\linewidth]{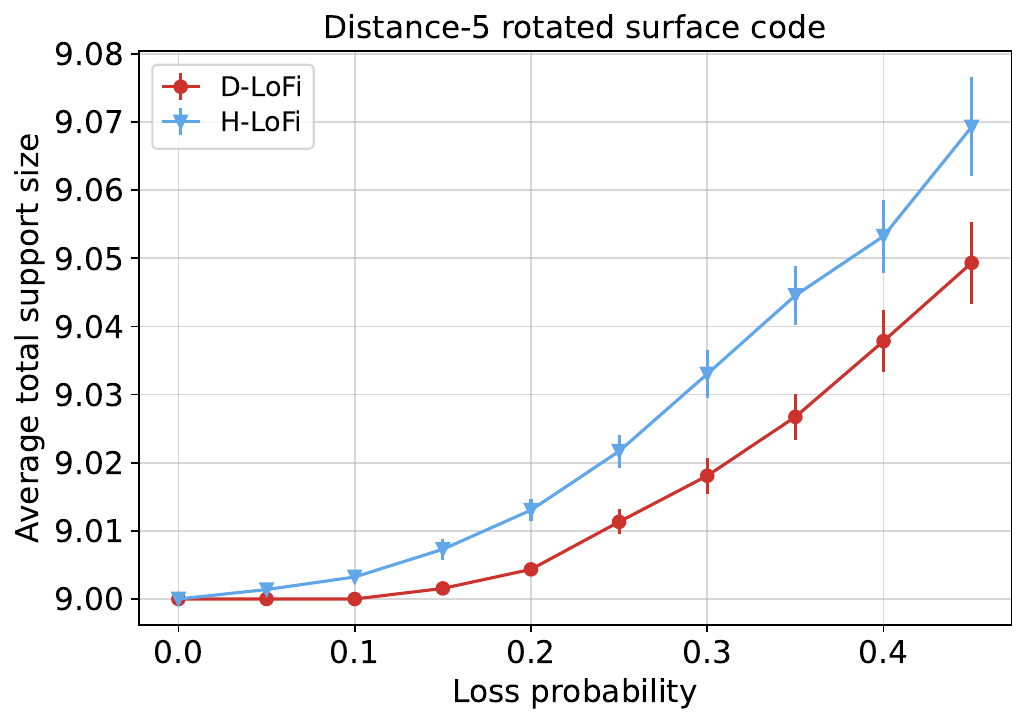}
        \subcaption{}
    \end{minipage}
    \begin{minipage}{0.49\textwidth}
        \includegraphics[width=0.95\linewidth]{sup_size_d5_p01.pdf}
        \subcaption{}
    \end{minipage}
    \caption{Here, we benchmark the runtimes and minimisation performance of D-LoFi and H-LoFi. We plot the averages over 5000 sampled loss configurations from the experiment run to obtain Fig.~\ref{fig: sc_thres}, along with the standard error of the mean (SEM). The SEM is sometimes too small to be visible. (a) shows the average runtime for the distance-5 rotated surface code with increasing loss probability. In (b), we plot the runtime with increasing code distance for a fixed loss probability ($p=0.1$). The runtime of our algorithms decreases with loss probability and increases with the code distance. (c) shows that the average size of the total support $|\supp(X)\cup \supp(Z)|$ obtained from the two algorithms increases with the loss probability, up to the threshold $p=0.5$. Here, only cases in which both algorithms find a $g$-SPF pair are considered. (d) shows the relative total support size compared to the no-loss case with increasing code distance at a fixed loss probability ($p=0.15$). We observe that the increase in the optimal support size (from D-LoFi) is very small with the code distance. This suggests that the effect of loss on the size of the optimal $g$-SPF pair becomes weaker for larger-distance rotated surface codes. Moreover, H-LoFi manages to attain the optimal solution (given by D-LoFi) in the vast majority of cases.}
    \label{fig: sc_thres_d5}
\end{figure*}
 In Fig.~\ref{fig: sc_thres}, we show the $g$-SPF threshold for the rotated surface code, with the value of $g$ fixed to 1. We consider code distances from 3 to 9. For each loss probability, we sample at least $N_L=5000$ loss configurations and test whether logical operators satisfying $g$-SPF can be found.  ($N_L$ is larger than 5000 for some probabilities and distances in order to keep the statistical error as small as possible.)  We plot the rate of success given by Eqn.~\eqref{eq:P_p_estimator} as a function of loss probability. The results from D-LoFi and H-LoFi are shown in Fig.~\ref{fig: sc_thres}(a) and Fig.~\ref{fig: sc_thres}(b), respectively. Both our algorithms validate the existence of a $g$-SPF threshold near $p=1/2$ for the rotated surface code, notably for $g=1$, which is a stricter statement than Theorem~\ref{thm:threshold_surface_informal}. Even though H-LoFi is not exact, we observe that H-LoFi recovers very similar success rates to D-LoFi and can thus be applied as a fast alternative for the analysis of localisation in stabiliser codes. 

We compare the runtime of D-LoFi and H-LoFi in Fig.~\ref{fig: sc_thres_d5}(a) and Fig.~\ref{fig: sc_thres_d5}(b). We average the runtimes over the loss configurations sampled for the creation of Fig.~\ref{fig: sc_thres} and compute the average runtimes for each loss probability and code size. In Fig.~\ref{fig: sc_thres_d5}(a), we compare the average runtime of the two algorithms for increasing loss probability.  We observe that H-LoFi is faster than D-LoFi for low loss probabilities.  

For a small range of high loss probabilities above the threshold $p=1/2$, D-LoFi is faster, contrary to what one would expect.
Moreover, we observe in Figures~\ref{fig: sc_thres}(a) and (b) that H-LoFi does not reach the same success rates as D-LoFi for small codes in the low probability regime. These two phenomena may have the following reasons: The pre-processing is the exact same for both algorithms and aborts the entire algorithm early and outputs failure if logical information has been destroyed. Above the erasure decoding threshold of $p=1/2$, this pre-processing dominates, which is why the behaviours of both algorithms are almost indistinguishable. Hence, H-LoFi needs to be faster than D-LoFi for a small number of instances to be faster on average. Additionally, at high probabilities more stabiliser generators are lost and the local structure of the code may not be preserved very well by the pre-processing. Very small codes, such as the distance-3 code, are sensitive to this disturbance of the local structure already at small probabilities. The BP-OSD decoder from Refs.~\cite{Roffe_LDPC_Python_tools_2022,roffeDecodingQuantumLowdensity2020} may struggle to converge due to this arbitrary non-local structure, whereas CP-SAT may be able to handle arbitrary stabiliser tableau structures better. This could be why H-LoFi does not yield the same success rates as D-LoFi for very small codes and becomes slower for high loss probabilities regardless of code size. A more careful manual implementation of the decoder instead of relying on the ldpc python package may be needed for more robust behaviour of H-LoFi in these regimes.

In Fig.~\ref{fig: sc_thres_d5}(b), we plot the runtime of our algorithms as a function of code distance. We fix the loss probability at 0.1. We observe that the runtime of D-LoFi increases fast with the code distance. On the other hand, the heuristic H-LoFi algorithm has a more favourable scaling compared to D-LoFi, demonstrating that it serves as a suitable alternative to the exact but slow D-LoFi. In order to evaluate the differences between solutions found by D-LoFi and the ones found by H-LoFi, we plot the average total support size of the solution pair $(X,Z)$ , i.e. $|\supp(X) \cup \supp(Z)|$, for increasing loss probability in Fig.~\ref{fig: sc_thres_d5}(c). Note that the total support size is an integer, meaning the non-integer numbers on the $y$-axis indicate that the total support size is fluctuating between \emph{integer} values between 9 and some large integer. We fix the distance of the rotated surface code to be 5.
  We only consider loss configurations where both our algorithms find solutions to $g$-SPF. We observe that the average support size increases with the loss probability for both algorithms, meaning at higher loss probabilities one finds measurement patterns with increasingly large support. H-LoFi diverges from D-LoFi at higher loss probabilities, but only to a very small extent. In Fig.~\ref{fig: sc_thres_d5}(d) we fix the loss probability $p$ at 0.15 and increase the distance. We plot the relative increase in total support size compared to the no-loss case. That is, we find a measurement pattern $(X_0,Z_0)$ at loss probability $p=0$ with minimal support and divide the total support $|\supp(X) \cup \supp(Z)|$ for each measurement pattern found at loss probability $p=0.15$ by $|\supp(X_0) \cup \supp(Z_0)|$. We observe that it becomes harder for H-LoFi to find the optimal solutions at higher distances and the gap to the optimal solution seems to grow exponentially, albeit with a very small exponent. We conclude that the H-LoFi algorithm provides a close approximation of D-LoFi, which finds the exact optimum. Therefore, H-LoFi is a suitable alternative to D-LoFi. Moreover, when one is not interested in the SPF pair with minimal support, for example when computing a threshold, H-LoFi constitutes a fast and accurate method, as demonstrated in Fig.~\ref{fig: sc_thres}.

\begin{figure*}[t]
    \begin{minipage}{0.49\textwidth}
        \includegraphics[width=0.95\linewidth]{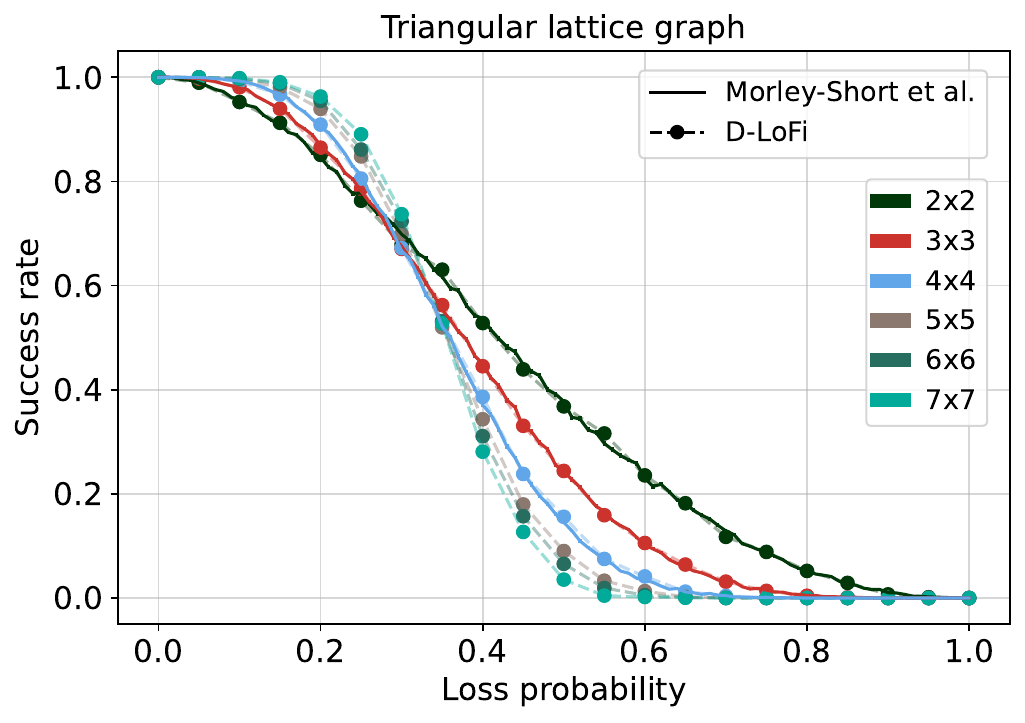}
        \subcaption{}
    \end{minipage}
    \begin{minipage}{0.49\textwidth}
        \includegraphics[width=0.95\linewidth]{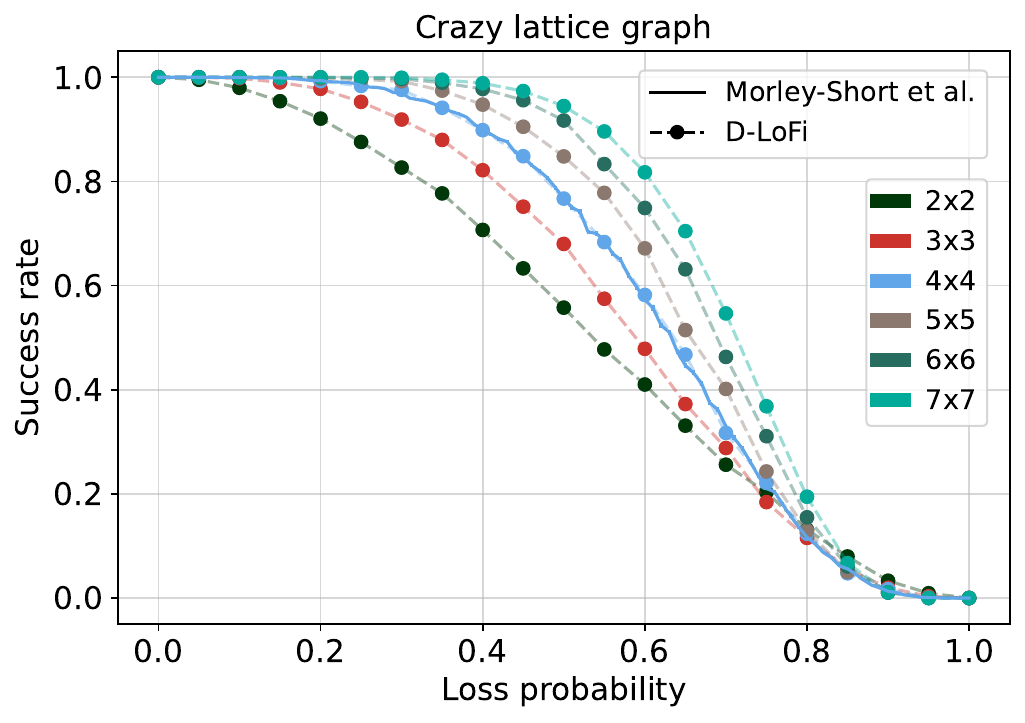}
        \subcaption{}
    \end{minipage}
    \caption{Threshold behaviour comparison of graph codes in the presence of loss. We sample 2000 loss configurations for each loss probability according to the i.i.d. single-qubit loss probability $p$. We plot the success rate of Target-SPF for the triangular lattice and crazy lattice graph codes depicted in Fig.~\ref{fig:channels} for increasing losses in {(a)} and {(b)} respectively. We consider the lattices of shape $l \times l$, and we increase the value of $l$ from 2 to 7. We also plot the data accompanying Ref.~\cite{morley-shortLosstolerantTeleportationLarge2019}, which is only available for lattice sizes up to $l=4$ for the triangular lattice graph code and for the single lattice size of $l=4$ for the crazy graph code. For the available data, we successfully reproduce similar results using the D-LoFi algorithm. We go further and plot the success rate of Target-SPF for lattices with size up to $7 \times 7$. The target qubits are the last qubits with respect to the indexation given in Fig.~\ref{fig:channels}. The SEMs are too small to be visible and have been omitted.}
    \label{fig: gc_thres_comparison}
\end{figure*}

We continue by comparing D-LoFi for Target-SPF to the algorithm from Ref.~\cite{morley-shortLosstolerantTeleportationLarge2019}.
In Fig.~\ref{fig: gc_thres_comparison} we recreate the convergence towards a loss threshold for Target-SPF that was demonstrated in Ref.~\cite{morley-shortLosstolerantTeleportationLarge2019}. We consider the triangular and crazy graph code (see Fig.~\ref{fig:channels}), and plot our results created using the D-LoFi algorithm together with the data obtained from Ref.~\cite{morley-shortLosstolerantTeleportationLarge2019}. As discussed in Section~\ref{sec:graph_codes}, the graph codes from Ref.~\cite{morley-shortLosstolerantTeleportationLarge2019} are parametrised by $l$, plus one input qubit and one target qubit. Ref.~\cite{morley-shortLosstolerantTeleportationLarge2019} assumes in its analysis that neither of these qubits is lost; therefore, we assume the same for these plots. 

We demonstrate in Fig.~\ref{fig: gc_thres_comparison} that D-LoFi yields the same success rates as in Ref.~\cite{morley-shortLosstolerantTeleportationLarge2019}, confirming the validity of our method. We run the D-LoFi algorithm for codes with lattices of sizes up to $7 \times 7$ (51 physical qubits). On the other hand, the data from Ref.~\cite{morley-shortLosstolerantTeleportationLarge2019} only goes up to the $4 \times 4$ code (18 physical qubits).

\begin{figure}[h!]
\hspace{0.225\textwidth}
    \begin{minipage}{0.99\textwidth}
        \includegraphics[width=0.5\linewidth]{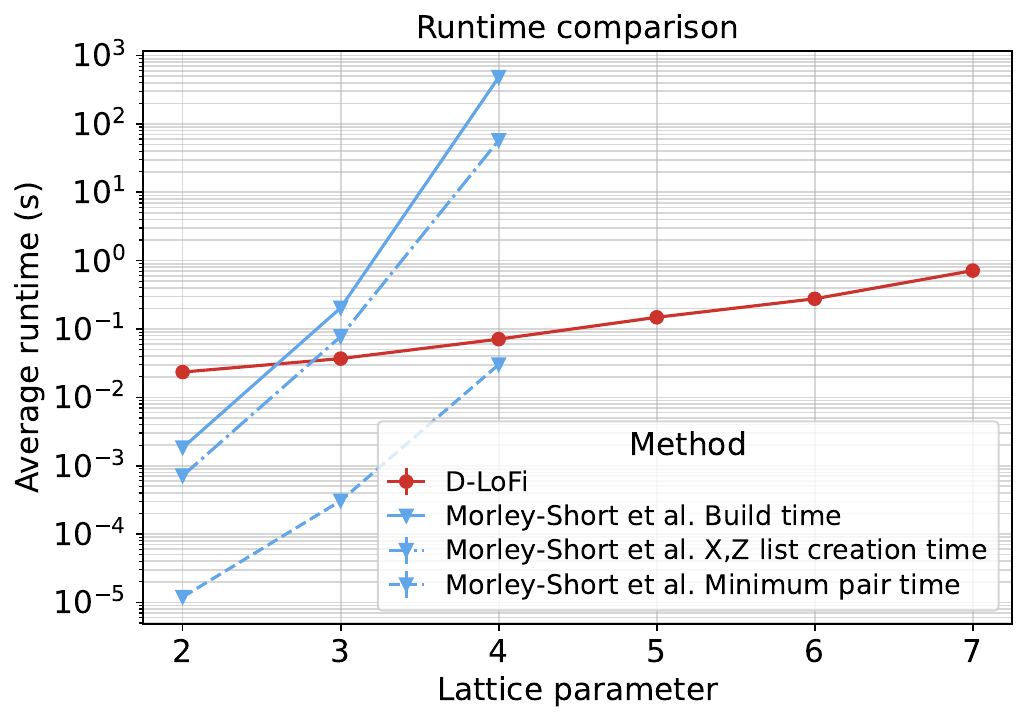}
    \end{minipage}
    \caption{Here, we plot the runtime comparison of D-LoFi with the algorithm from Ref.~\cite{morley-shortLosstolerantTeleportationLarge2019} (labelled as \emph{Morley-Short et al.}) when searching for an SPF pair with minimal support. We measured the build time, which is the time to list all non-trivial stabilisers (see Section~\ref{sec:SPFprevious}) only once, as the algorithm is deterministic and because we only had limited computational resources. We do the same for the XZ list creation time, which is the time to create a list of all possible combinations of non-trivial $(X,Z)$ that fulfil the Target condition. We then sample either all possible loss configurations (for small graphs) or 1000 unique loss configurations and compute the average time it takes to find the SPF pair with minimal support, along with the standard error of the mean, which is too small to be visible. Due to the long runtime of the algorithm from Ref.~\cite{morley-shortLosstolerantTeleportationLarge2019}, we could only scale the crazy lattices to the $4 \times 4$ lattice. The build time exceeded the imposed wall time limit of 60,000 seconds ($\sim$ 16 hours and 40 minutes) for larger lattices. We compare the runtimes to the average runtime of D-LoFi for the same loss configurations. We observe that our algorithm has a significantly shorter runtime and better scaling with increasing lattice size, allowing us to scale up to the $7  \times 7$ lattice. We ran these simulations on the DelftBlue supercomputer on one CPU core with 1 GB of memory.}
    \label{fig: ms_rt_comp}
\end{figure}

We compare the runtime for the algorithm from Ref.~\cite{morley-shortLosstolerantTeleportationLarge2019} with D-LoFi for finding minimum logical support size while solving Target-SPF. We consider the crazy graph code with increasing lattice size.  The algorithm from Ref.~\cite{morley-shortLosstolerantTeleportationLarge2019} consists of (1) exhaustively listing all possible non-trivial stabilisers, where they take advantage of the graph structure of the code, and (2) creating a list of all possible non-trivial logical operator combinations $(X,Z)$, sorted by support size. In step (2) they create logical representatives fulfilling the Target condition by multiplying non-trivial stabilisers with some starting representative and checking whether the resulting logical is non-trivial. (For more details we refer to the GitHub repository~\cite{sammorley_short_spf_2019}.) In Fig.~\ref{fig: ms_rt_comp}, we call the time to perform step (1) the `build time' and the time to perform step (2) the `XZ list creation time'. It should be noted that the sorting process that happens in step (2) can already be considered part of the search for a minimal SPF pair. We then sample loss configurations and search the list for an SPF pair with no support on the lost qubits using a primitive loop that starts with the lowest-weight pairs in the list. We call this the `minimum pair time' and plot its average in Fig.~\ref{fig: ms_rt_comp}. The search could be optimised, but as it is by far the fastest part of the entire method, our focus is on the creation of the look-up table. We compare the runtimes to the runtime of D-LoFi averaged over the same unique loss configurations. 
We ran all code on the DelftBlue supercomputer on one CPU core with 1 GB of memory.

We observe that both the build time and the XZ list creation time from Ref.~\cite{morley-shortLosstolerantTeleportationLarge2019} scale orders of magnitude worse than the average runtime of the D-LoFi algorithm. Moreover, due to the unfavourable scaling of their algorithm, we could only run these simulations for lattices of size up to $4 \times 4$, as simulations for larger codes exceeded the 60,000 seconds ($\sim$ 16 hours and 40 minutes) wall time limit. In order to simulate larger codes, one could increase the number of cores used on the supercomputer. However, we are interested in a simple comparison between D-LoFi and the algorithm from Ref.~\cite{morley-shortLosstolerantTeleportationLarge2019} and hence use the same computational resources for both of them.
One argument for the exhaustive listing of all possible measurement patterns $(X,Z)$ in Ref.~\cite{morley-shortLosstolerantTeleportationLarge2019} is that it creates a look-up table which can then enable very fast identification of measurement patterns in an experimental setting. However, our algorithms can create look-up tables on the fly. That is, one can simply sample loss configurations and store each measurement pattern, creating a cache. As D-LoFi is orders of magnitude faster than the exhaustive method from Ref.~\cite{morley-shortLosstolerantTeleportationLarge2019}, running D-LoFi for thousands of instances could still be faster than the creation of the look-up table in Ref.~\cite{morley-shortLosstolerantTeleportationLarge2019}. Moreover, our heuristic algorithm H-LoFi could be adapted for Target-SPF and can easily be randomised to yield different SPF solutions for the same loss configuration. D-LoFi technically need not be formulated as a minimisation problem, but rather as a satisfiability problem, for which CP-SAT can be randomised. Moreover, CP-SAT from Ref.~\cite{ortools,cpsatlp} can also exhaustively list all solutions. All of this is to say that our methods can also be used to create a large number of valid measurement patterns. On top of that, the search for a measurement pattern fulfilling the Target- and Lost-Qubits condition in the list of all combinations of non-trivial logicals is not necessarily a small task, as also evidenced by the scaling of the minimum pair time in Fig.~\ref{fig: ms_rt_comp}.

Lastly, Fig.~\ref{fig: tvsg_spf_gc} shows the success rate of finding SPF pairs for increasing loss probability for Target-SPF (fixed target) vs $g$-SPF, where $g=1$ (no fixed target). For Target-SPF, we assume that the target qubit cannot be lost (same as Ref.~\cite{morley-shortLosstolerantTeleportationLarge2019}). In the case of $g$-SPF, we assume that all the code qubits can be lost. We observe that relaxing the Target condition to requiring $X$ and $Z$ to anti-commute on any arbitrary qubit yields higher success rates, demonstrating that the a priori designation of a target qubit may sometimes be sub-optimal.

\begin{figure*}[t]
    \begin{minipage}{0.49\textwidth}
        \includegraphics[width=0.95\linewidth]{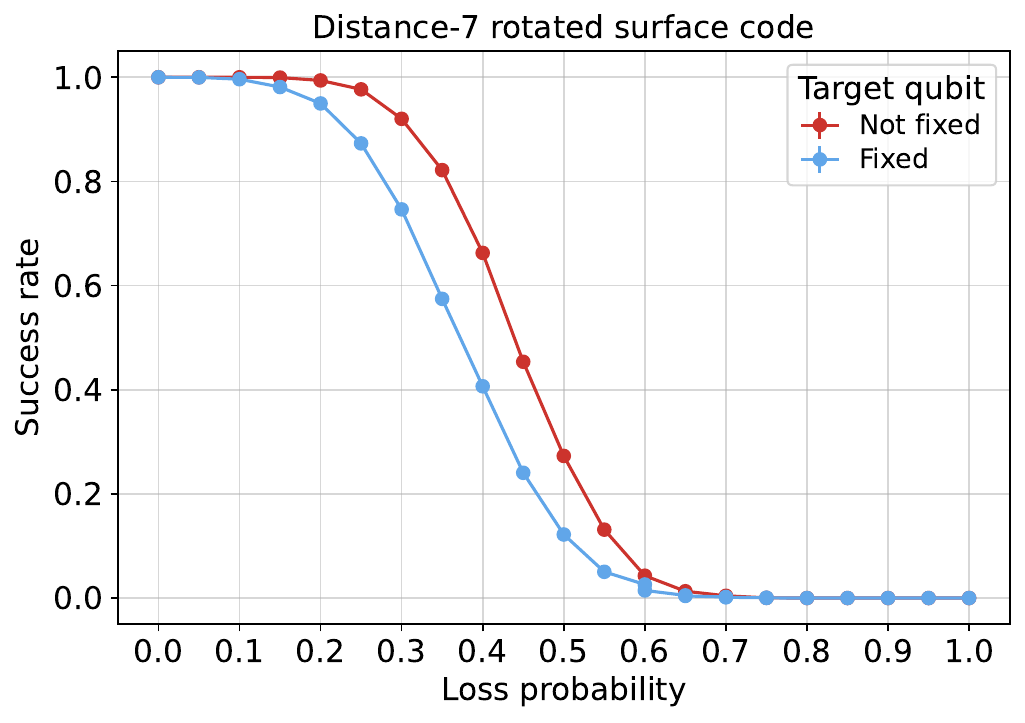}
        \subcaption{}
    \end{minipage}
    \begin{minipage}{0.49\textwidth}
        \includegraphics[width=0.95\linewidth]{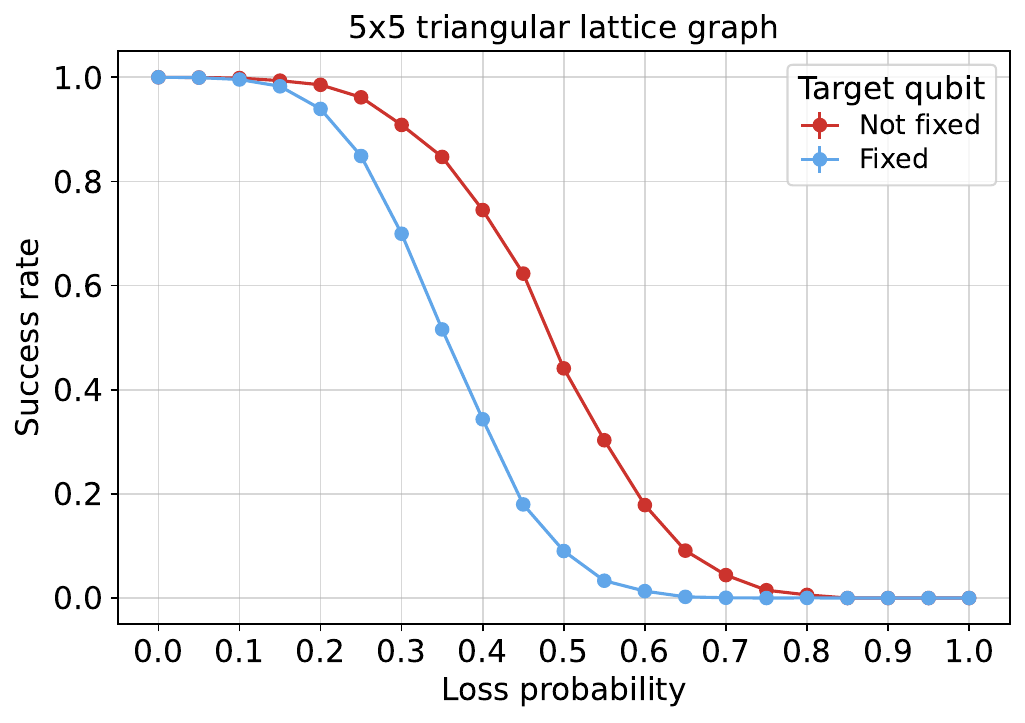}
        \subcaption{}
    \end{minipage}
    \caption{Here, we compare the success rate of the D-LoFi algorithm for Target-SPF and $g$-SPF with $g=1$. In (a), we sample 3500 loss configurations for each loss probability, while in (b) we sample 2000 loss configurations for each loss probability. The SEMs are too small to be visible. In the case of Target-SPF we assume that the input and the target qubit are never lost. We consider the distance-7 rotated surface code in (a), and the $5 \times5$ triangular lattice graph code in (b). For both codes we observe that choosing a fixed target qubit reduces the average success rate compared to letting it be arbitrary. }
    \label{fig: tvsg_spf_gc}
\end{figure*}

\section{Conclusion and outlook}\label{sec:concl}

In this work, we have provided a broad analysis of the stabiliser path finding (SPF) problem and associated localisation phenomena. We have proved that the planar surface code has a localisation threshold of $p=1/2$ via $g$-SPF for constant $g$, meaning that for sufficiently large physical codes one can localise information onto a constant number of qubits via single-qubit Pauli measurements with a probability converging to one. Moreover, we have designed and implemented two algorithms for Target- and $g$-SPF: (1) A deterministic algorithm (D-LoFi) based on the formulation of SPF as an integer linear program. (2) A heuristic algorithm (H-LoFi) based on formulating SPF as a decoding problem. For a given code and a loss configuration $L$, both algorithms directly find a measurement pattern $(X,Z)$ (H-LoFi only for $g$-SPF).
 
We use our H-LoFi and D-LoFi algorithms to solve $g$-SPF for the rotated surface code with code distances up to 9 (81 physical qubits). We show the existence of loss thresholds for the rotated surface code for $g=1$, consistent with our analytical result in Theorem~\ref{thm:threshold_surface_informal}. H-LoFi, despite not being an exact algorithm, recovers similar success rates to the deterministic D-LoFi. Furthermore, we show that H-LoFi scales favourably with the system size compared to D-LoFi. 
We also demonstrate the use of the D-LoFi algorithm to solve Target-SPF for graph codes. We compute the SPF success rates for two graph codes, namely the triangular and the crazy graph code considered in Ref.~\cite{morley-shortLosstolerantTeleportationLarge2019}. We first compare the results from the D-LoFi algorithm with the data from Ref.~\cite{morley-shortLosstolerantTeleportationLarge2019,sammorley_short_spf_2019} and reproduce similar loss thresholds. Moreover, we extend the analysis to graph codes of larger size (up to 51 physical qubits) beyond the 18 physical qubits that were considered previously. Finally, by comparing the runtimes of the algorithm from Ref.~\cite{morley-shortLosstolerantTeleportationLarge2019,sammorley_short_spf_2019} and of the D-LoFi algorithm for Target-SPF we show that our algorithm is substantially faster and scales better.

There are many interesting open directions to be investigated further. 
On the theoretical side, an interesting question is whether our proof for the existence of a localisation threshold for the planar surface code (i.e. Theorem~\ref{thm:threshold_surface_informal}) can be extended to general codes. To be precise, it would be interesting to see which properties of a code make it well suited for localisation via $g$- and Target-SPF. The proof in Section~\ref{sec:threshold_surfacecode} uses support intersection as a proxy for qubit-wise anti-commutation. However, heavily overlapping operators may still anti-commute qubit-wise on a small number of qubits, meaning intersection is a loose proxy. Hence, a generalisation of the proof possibly using a notion of disjointness that is aware of qubit-wise commutation is needed. 

The authors of Ref.~\cite{PhysRevA.110.052617} consider localisation via LOCC of quantum information \emph{in the absence of loss} and derive a necessary and sufficient condition for the existence of such a protocol. The condition and the efficiency of the algorithm construction rely on the expressibility of a \emph{pure} stabiliser state as a graph state and thus cannot be trivially applied to the generally mixed state after loss. It would be interesting to see whether one can find a necessary and sufficient condition for localisation via LOCC after loss and whether it is related to the existence of an SPF pair. In addition, it would be interesting to analyse whether there exist fundamental limits on the localisation success rate for a code undergoing qubit loss, similar to  Shannon's channel capacity bound for communication over noisy channels~\cite[Chs.~11 and 12]{nielsenQuantumComputationQuantum2010}. Another future direction could be the consideration of multiple logical qubits in $g$-SPF and Target-SPF. 

Furthermore, we here did not analyse the complexity of SPF as a decision problem, but we conjecture it is NP-complete given its similarity to the problem of deciding whether a given  stabiliser state can be transformed into a set of Bell pairs on specific vertices using only single-qubit Clifford operations, single-qubit Pauli measurements and classical communication~\cite{Dahlberg2020transforming}. In Ref.~\cite{PhysRevA.110.052617}, the authors also construct an efficient algorithm to find an LOCC localisation protocol. We suspect that no such efficient algorithm can be found after loss.  
Practically, SPF has previously been utilised in the loss-tolerance analysis of graph codes in Ref.~\cite{bellOptimizingGraphCodes2023}. There, the authors use SPF to search for graph codes well suited for teleportation, fusion and other applications. Considering that our methods demonstrate a considerable speed-up, one could use them to include larger graph codes in the search. It should be noted that Ref.~\cite{bellOptimizingGraphCodes2023} discusses heralded as well as unheralded loss. In the case of unheralded loss one only detects qubit loss when one attempts to measure the respective qubit. To model unheralded loss, one needs to model these single-qubit measurements. However, the pre-processing can straightforwardly be extended to take previous measurements as an input. (For a more detailed explanation see Appendix~\ref{app:measurements}.) We would like to extend this work by trying to characterise optimal measurement policies for SPF in the case of unheralded loss, both numerically and analytically. To find good strategies for localisation after loss, one could try to find analogies to the findings in Ref.~\cite{goodenough2026optimalfusionstrategiesquantum}, where the authors discuss optimal fusion strategies in the case of fusion measurements that can fail. Additionally, combining the work in Ref.~\cite{goodenough2026optimalfusionstrategiesquantum} with our work here may lead to an interesting future direction of finding fusion failure strategies in the presence of photon loss. 
On the numerical side, one could for example apply reinforcement learning to find good SPF strategies, similarly to Refs.~\cite{PhysRevApplied.23.034048,doherty2026faststabilizerstatepreparation}. When searching for good policies in the case of unheralded loss, creating a large number of measurement patterns may be advantageous for the fast evaluation of some heuristic for choosing the next measurement. For example, in Ref.~\cite{morley-shortLosstolerantTeleportationLarge2019} the authors evaluate how many measurement patterns $(X,Z)$ a specific qubit is part of. Ref.~\cite{bellOptimizingGraphCodes2023} also creates look-up tables and performs exhaustive tree-searches based on them. Even though our algorithms find one measurement pattern per loss configuration, they can, as discussed in Section~\ref{sec:numerics}, also be used to create a large number of valid measurement patterns on the fly. That is, one can create look-up tables while already solving SPF. Hence, one need not sacrifice speed and can use our algorithms to systematically search for strategies and suitable codes for SPF in the case of unheralded loss.

To keep the comparison with the methods from Ref.~\cite{morley-shortLosstolerantTeleportationLarge2019} fair, our proposed algorithms have not been optimised for performance and serve as proof-of-principle implementations.  Multi-threading and multi-processing could be used to speed up both our algorithms further. More specifically, the D-LoFi algorithm could be sped up by making better use of all features in CP-SAT~\cite{ortools,cpsatlp}, which we used to solve the ILP. Furthermore, the heuristic algorithm could be improved by encoding the anti-commutation condition directly as a constraint into the decoder from Refs.~\cite{Roffe_LDPC_Python_tools_2022,roffeDecodingQuantumLowdensity2020} that was used for H-LoFi, instead of treating it like a black box, as we have done in this work. Target-SPF could also be directly encoded this way. 

When building experimental loss-tolerant photonic systems, such as the one in Ref.~\cite{qec_in_photonic}, our work would allow for faster localisation than previously possible. Further runtime reductions could allow for the online application of SPF for localisation of information in experimental systems undergoing loss. As SPF can be used for teleportation over a lossy channel or fault-tolerant fusion measurements~\cite{bellOptimizingGraphCodes2023}, our work could lead to higher success rates with respect to these applications.

\section*{Acknowledgements}
The authors thank Francisco Machado and Kenneth Goodenough for helpful discussions and interesting suggestions. All authors acknowledge funding from the Dutch Research Council (NWO) through the Quantum Software Consortium (NWO Zwaartekracht Grant No.024.003.037). JH acknowledges funding through a Veni grant (grant No.VI.Veni.222.331).

\section*{Code availability}
The accompanying code can be found at \url{https://github.com/arr0w-hs/graph_code_decoding}.

\printbibliography
\appendix

\section{D-LoFi and H-LoFi pseudocode}\label{app:algos}
Here we give the pseudo-code for the pre-processing described in Section~\ref{sec:prepro} and D- and H-LoFi.
\begin{algorithm}[H]
\caption{Cleaning the logicals and the stabiliser tableau}\label{algo:prepro}
\KwIn{$r \times (2n)$ stabiliser tableau $T$, $1 \times (2n)$ representatives $X \in [X]$ and $Z \in [Z]$, set of lost qubits $L \subseteq Q$}
\KwOut{Cleaned and reduced $r' \times 2(n-|L|)$ tableau $T'$, cleaned and reduced $1 \times 2(n-|L|)$ representatives $X'$ and $Z'$}
Solve Eqn.~\eqref{eq:cleanlogical} for $X$ and for $Z$: obtain $c_X$, $c_Z$ or no solution\;
\If{no solution}{
    \Return{failure}\;
}
$X' \leftarrow (X + c_X T)$, $Z' \leftarrow (Z + c_Z T)$\;
Solve Eqn.~\eqref{eq:clean_T}: obtain a basis $C_L$ of $\mathcal{C}_L$\;
$T' \leftarrow (C_L T)[Q-L]$\;
\Return{$T'$, $X'[Q-L]$, $Z'[Q-L]$}\;
\end{algorithm}

\begin{algorithm}[H]
\caption{D-LoFi for Target-SPF and $g$-SPF}\label{algo:det}
\KwIn{$r \times 2n$ full-rank stabiliser tableau $T$ with $k=1$,
set of lost qubits $L \subseteq Q$,
and either a target qubit $t \notin L$ (Target-SPF) or an integer $g$ ($g$-SPF)}
\KwOut{SPF pair $(X^\star,Z^\star)$ of minimal support, or failure}

$(X,Z) \gets$ Algorithm~\ref{algo:sympl}$(T)$
\tcp*{$k=1$, so a single symplectic pair is returned}

\tcp{Pre-processing: Lost-Qubits condition (Definition~\ref{def:Lemptyset})}
$(T',X_0,Z_0) \gets$ Algorithm~\ref{algo:prepro}$(T,X,Z,L)$\;
\lIf{failure}{\Return{failure}}

\tcp{Set up the ILP over the $n' = n-|L|$ remaining qubits, re-indexed as $0,\ldots,n'-1$; $t$ denotes the new index of the target qubit}
Declare binary variables $B_x,B_z \in \mathbb{F}_2^{1 \times r'}$ and set
$X = X_0 + B_xT'$, $Z = Z_0 + B_zT'$\;
\tcp{All sums of symplectic products below are taken over $\mathbb{Z}$}
\uIf{Target-SPF}{
  Add $\langle X[t],Z[t]\rangle = 1$ and $\sum_{i \neq t}\langle X[i],Z[i]\rangle = 0$
  \tcp*{Target condition, Definition~\ref{def:tanti}}
}
\Else(\tcp*[f]{$g$-SPF}){
  Add $\sum_{i=0}^{n'-1}\langle X[i],Z[i]\rangle \leq g$
  \tcp*{$g$-condition, Definition~\ref{def:gcondition}}
}
Set the objective to minimise $|\supp(X) \cup \supp(Z)|$\;
Linearise the mod-2 sums and the products of binary variables by introducing auxiliary variables and linear constraints\;

Solve the ILP with an ILP solver\;
\lIf{infeasible}{\Return{failure}}
\Return{$X^\star = X_0 + B_x^\star T'$, $Z^\star = Z_0 + B_z^\star T'$}\;
\end{algorithm}

\begin{algorithm}[H]
\caption{H-LoFi for $g$-SPF}\label{algo:gspf_heuristic}
\KwIn{$r \times 2n$ stabiliser tableau $T$; lost qubits $L$; integer $g$; penalty $q^-$; iterations $N$}
\KwOut{A $g$-SPF pair $(X',Z')$, or failure}
Find a symplectic basis $\{(X_i,Z_i)\}_i$ of $\ker(T\Omega)/\mathcal{R}(T)$ via Algorithm~\ref{algo:sympl}\;
$(X,Z) \gets (X_0,Z_0)$\tcp*{there is only one logical qubit}
Apply pre-processing (Algorithm~\ref{algo:prepro}) to $(T,X,Z,L)$\;
\If{logical information loss}{\Return{failure}}
Obtain reduced $r' \times 2(n-|L|)$ tableau $T'$ and cleaned logicals $(X',Z')$\;
Find a symplectic basis of $\ker(T'\Omega)/\mathcal{R}(T')$ containing $(X',Z')$ via Algorithm~\ref{algo:sympl}\;
$T_0 \gets T'$\;
\ForEach{pair $(X_j,Z_j) \neq (X',Z')$}{
    $T_0 \gets \begin{pmatrix} T_0\\ X_j\\ Z_j \end{pmatrix}$ \tcp*{Eqn.~\eqref{eq:append_sympl}}
}
\For{$\mathrm{iter} \gets 1$ \KwTo $N$}{
    Set weights $q$: equal if $\mathrm{iter}=1$, else drawn at random\;
    \tcp{Part 1: replace $Z'$ with a short representative}
    $T_Z \gets \begin{pmatrix} T_0\\ X' \end{pmatrix}$\;
    Solve Eqn.~\eqref{eq:decoding_short_Z_threeblock} on $R_{2\mapsto 3}(T_Z)$ with syndrome $(0,\dots,0,1)^{\top}$ and weights $q$; obtain short $Z'$\;
    \If{no solution}{\textbf{continue}}
    \tcp{Part 2: find a short $X'$ obeying the $g$-condition}
    $T_X \gets \begin{pmatrix} T_0\\ Z' \end{pmatrix}$\;
    \ForEach{$i \in \supp(Z')$}{
        \ForEach{$P \in \{X,Y,Z\}$ with $P \neq Z'[i]$}{
            $q^P[i] \gets q^-$ \tcp*{Eqn.~\eqref{eq:q_punish}}
        }
    }
    Solve Eqn.~\eqref{eq:decoding_short_Z_threeblock} on $R_{2\mapsto 3}(T_X)$ with syndrome $(0,\dots,0,1)^{\top}$ and weights $q$; obtain $X'$\;
    \If{$X'$ found \textbf{and} $\alpha(X',Z') \leq g$}{\Return{$(X',Z')$}}
}
\Return{failure}\;
\end{algorithm}

\section{Miscellaneous algorithms}\label{app:misc}

Here we explain some miscellaneous algorithms. First, we will discuss how to construct a symplectic basis for a stabiliser group $\Scurvy$. Then, we will explain how to add Pauli measurements to the pre-processing of Section~\ref{sec:prepro}.

\subsection{Construction of a symplectic basis}

Let us define a symplectic basis:
\begin{definition}[Symplectic basis]
Let $\Scurvy \subseteq \mathcal{P}_n$ be a stabiliser group defining $k$ logical qubits. A symplectic basis $\{(X_i,Z_i)\}_i$ of $\mathcal{N}(\Scurvy)/\Scurvy$ consists of $k$ pairs $(X_i,Z_i)$, where each $X_i$, $Z_i$ $\in \mathcal{N}(\Scurvy)-\Scurvy$ and $[X_i,Z_j]=0$ $\forall i \neq j$, $\{X_i,Z_i\}=0$ and $[X_i,X_j]=0$ and $[Z_i,Z_j]=0$ $\forall$ $i,j$.
\end{definition}
In other words, the symplectic basis generates a group isomorphic (up to phases) to the logical Pauli operators $X$ and $Z$ on $k$ qubits. To find a symplectic basis, we can first find a basis $\{b_i\}$ for $\ker(T\Omega)/\mathcal{R}(T)$. However, it is not guaranteed that $\{b_i\}$ fulfil the desired anti-commutation relations. We can create a symplectic basis iteratively.

Given any initial $v \in \{b_i\}$ and $w \in \{b_i\}$ with $\langle v,w \rangle=1$, we find a linearly independent set $\{\tilde{b}_i\}$ that commutes with both $v$ and $w$ element-wise the following way:

\begin{itemize}
\item Take some $b_i \in \{b_i\}.$ Check whether $\langle b_i,v\rangle=0$. If yes, proceed with the next $b_i$. If not, map $b_i \mapsto b_i + w$. Obtain $\{\tilde{b}_i\}$.
\item Do the same but this time for $\{\tilde{b}_i\}$ and $w$, mapping $\tilde{b}_i \mapsto \tilde{b}_i + v$ where necessary.
\item Stack all $\tilde{b}_i$ as rows of a matrix and row reduce.
\end{itemize}

We can then proceed with the remaining elements in $\{\tilde{b}_i\}$ analogously. 

We summarise the construction of a symplectic basis in Algorithm~\ref{algo:sympl}. It assumes that the input $T$ is full rank.
\begin{algorithm}[H]
\caption{Creating a symplectic basis}\label{algo:sympl}
\KwIn{$r  \times 2n$ full-rank stabiliser tableau $T$}
\KwOut{A symplectic basis $\{(X_i,Z_i)\}_{i=0}^{k-1}$}
$k \gets \big(\dim(\ker(T\Omega)/\mathcal{R}(T))\big)/2$\tcp*{$=n-r$}
Find a basis $\{b_1,\dots,b_{2k}\}$ for $\ker(T\Omega)/\mathcal{R}(T)$\;
$B \gets \{b_1,\dots,b_{2k}\}$\;
\For{$i \gets 1$ \KwTo $k$}{
    Pick any $X_i \in B$;\quad $B \gets B \setminus \{X_i\}$\;
    Pick any $Z_i \in B$ with $\langle X_i, Z_i\rangle = 1$;\quad $B \gets B \setminus \{Z_i\}$\;
    \ForEach{$b \in B$}{
        $b \gets b \;\oplus\; \langle b, Z_i\rangle\, X_i \;\oplus\; \langle b, X_i\rangle\, Z_i$\;
    }
    Stack the elements of $B$ as rows, row-reduce, keep it as the updated $B$\;
}
\Return{$\{(X_i,Z_i)\}_{i=1}^{k}$}\;
\end{algorithm}

 \subsection{Updating $T$ after measurements}\label{app:measurements}
Here we discuss how to update a stabiliser tableau $T$ after a Pauli measurement has been performed. This is needed to adapt our pre-processing from Section~\ref{sec:prepro} to SPF in the presence of unheralded loss.
A stabiliser group $\Scurvy$ is updated after a Pauli measurement $\pm M$, where the sign is determined by the outcome, by adding $\pm M$ to the stabiliser group and removing any $s \in \Scurvy$ that anti-commutes with $M$. Any representative $X \in [X]$ and $Z \in [Z]$ commuting with $M$ thus remains a valid logical representative~\cite{morley-shortLosstolerantTeleportationLarge2019,gottesmanStabilizerCodesQuantum1997}. We will now discuss how to perform this stabiliser group update via the symplectic representation. 

To be precise, when one performs a measurement of a Pauli operator $P \in \mathcal{P}_n-\Scurvy$ with respect to a state stabilised by $\Scurvy$, the group $\Scurvy$ is updated the following way~\cite{gottesmanStabilizerCodesQuantum1997} (neglecting the specification of measurement results as they are irrelevant for deciding whether an SPF pair exists):

\begin{itemize}
\item If $P \in \mathcal{N}(\Scurvy)-\Scurvy$: $P$ is added to $\Scurvy$. The number of generators of $\Scurvy$ is increased by one.
\item If $P \notin \mathcal{N}(\Scurvy)$: Any $s \in \Scurvy$ that anti-commutes with $P$ is removed while any $s$ that commutes with $P$ remains. $P$ is added to $\Scurvy$. The number of generators of $\Scurvy$ does not change.
\end{itemize}
In the following we assume that none of the measurements is an element of $\mathcal{N}(\Scurvy)$, i.e. each measurement anti-commutes with some stabiliser. Consider a generating set $\{m_i\}$ of the measurement set $M \subseteq \mathcal{P}_n$. We write $n_M:=|\{m_i\}|$.
To update $\Scurvy$, one can find a destabiliser basis $\{d_i\} \subseteq \Scurvy$ of $\{m_i\}$. That is, we must find a basis $\{d_i\}$ such that $\{d_i,m_i\}=0$ and $[d_i,m_j]=0$ for all $j \neq i$. Then, for any $g$ in some minimal generating set $\{g_i\}$ of $\Scurvy$, one does the following:
\begin{equation}\label{eq:destab_multiplication}
g \mapsto c(g):= g \prod_{i:\,\{m_i,g\}=0} d_i.
\end{equation}
This makes any $g$ commute with all measurements. The set $\Scurvy_c:=\{c(g) \mid g \in \{g_i\}\}$ of all updated generators is not a minimal generating set; it generates a group with $r-n_M$ independent generators. The full, new stabiliser group is generated by $\Scurvy_c$ and $\{m_i\}$.

In the symplectic representation, the destabilisers $d_i$ are given by the rows of the $n_M \times 2n$ matrix $D:=CT$, where $C \in \mathbb{F}_2^{n_M \times r}$ is a coefficient matrix. Let $M$ be the $2n \times n_M$ matrix whose $i$-th column is the symplectic representation of $m_i$ as a column vector. In order for the rows of $D$ to define a valid destabiliser basis, $C$ must fulfil:
\begin{equation}\label{eq:find_destab}
    CT\Omega M= I_{n_M},
\end{equation}
where $I_{n_M}$ is the $n_M \times n_M$ identity matrix.
Finding $\Scurvy_c$ corresponds to mapping $T$ to $T_c$:
\begin{equation}
    T_c :=T + (T\Omega M)D
\end{equation}
and row-reducing, obtaining a rank $r-n_M$ matrix. In the end, one removes all all-zero rows from the row-reduced $T_c$ and appends $M^T$ vertically.

Any vector $v \in \symplspace$ corresponding to an element of $[X]$ or $[Z]$ that anti-commutes with at least one of the measurements  can also be updated via:
\begin{equation}
    v\mapsto v+(v\Omega M)D.
\end{equation}
The result then commutes with all measurements.

\end{document}